\documentclass[journal,draftclsnofoot,12pt,onecolumn]{IEEEtran}
\usepackage{cite}  
\usepackage{color}
\usepackage{subfigure}
\usepackage{algorithm} 
\usepackage{algorithmic} 
\usepackage{multirow} 
\usepackage{amsmath} 
\usepackage{xcolor}

\usepackage{graphicx}
\usepackage{epstopdf}

\usepackage{amsfonts}
\usepackage{stfloats}
\usepackage{amsthm,amsmath,amssymb}
\usepackage{mathrsfs}
\usepackage{hyperref}
\usepackage{bm}
\usepackage{amsmath}

\newtheorem{lemma}{Lemma}
\newtheorem{proposition}{Proposition}

\ifCLASSINFOpdf

\else

\fi

\begin{document}
\title{Robust Joint Design for Intelligent Reflecting Surfaces Assisted Cell-Free Networks}

\author{Xie Xie, 
        Chen He, 
        Xiaoya Li, 
        Zhu Han,~\IEEEmembership{Fellow,~IEEE,}\\
        and Z. Jane Wang,~\IEEEmembership{Fellow,~IEEE}
\thanks{Xie Xie, Chen He, and Xiaoya Li are with the School of Information Science and Technology, Northwest University, Xi'an, 710069, China. Corresponding author: Chen He (email: chenhe@nwu.edu.cn).}
\thanks{Zhu Han is with the Department of Electrical and
Computer Engineering, University of Houston, TX, USA.}
\thanks{Z. Jane Wang is with the Department of Electrical and Computer Engineering, The University of British Columbia, Vancouver, BC V6T1Z4, Canada.}
}

\markboth{Submitted to IEEE Trans on Wireless Communications}
{Shell \MakeLowercase{\textit{et al.}}: Bare Demo of IEEEtran.cls for IEEE Journals}

\maketitle

\begin{abstract}
Intelligent reflecting surfaces (IRSs) have emerged as a promising economical solution to implement cell-free networks. However, the performance gains achieved by IRSs critically depend on smartly tuned passive beamforming based on the assumption that the accurate channel state information (CSI) knowledge is available, which is practically impossible. Thus, in this paper, we investigate the impact of the CSI uncertainty on IRS-assisted cell-free networks. We adopt a stochastic programming method to cope with the CSI uncertainty by maximizing the expectation of the sum-rate, which guarantees robust performance over the average. Accordingly, an average sum-rate maximization problem is formulated, which is non-convex and arduous to obtain its optimal solution due to the coupled variables and the expectation operation with respect to CSI uncertainties. As a compromising approach, we develop an efficient robust joint design algorithm with low-complexity. Particularly, the original problem is equivalently transformed into a tractable form, and then, the locally optimal solution can be obtained by employing the block coordinate descent method. We further prove that the CSI uncertainty impacts the design of the active transmitting beamforming of APs, but surprisingly does not directly impact the design of the passive reflecting beamforming of IRSs. It is worth noting that the investigated scenario is flexible and general, and thus the proposed algorithm can act as a general framework to solve various sum-rate maximization problems. Simulation results demonstrate that IRSs can achieve considerable data rate improvement for conventional cell-free networks, and confirm the resilience of the proposed algorithm against the CSI uncertainty.
\end{abstract}


\IEEEpeerreviewmaketitle

\section{Introduction}
In the past decade, the cell-free networks have been proposed as a promising technology to significantly improve the system performance compared with the traditional multi-cell system \cite{5594708}. In particular, cell-free networks advocate a more active treatment of interference, where multiple access points (APs) transmit data to multiple user equipment (UEs) cooperatively and simultaneously \cite{7917284,9354156,7827017}. 
On the other hand, the performance of the cell-free network critically depends on the APs with high hardware-cost and energy-consumption, which become a bottleneck to substantially improve its performance \cite{5594708}.

As a remedy, recently, intelligent reflecting surfaces (IRSs) have attracted a significant amount of attention as a potential economical solution for the future cell-free communication networks \cite{9110915}.
Generally, the basic function of an IRS is to smartly reconfigure the wireless propagation environment by reflecting the incident signal towards the desired spatial direction \cite{9086766}. Since an IRS composes of a large amount of phase shifters (PSs), each of which operates in a passive way without decoding and encoding operations, an IRS consumes much less power than an AP \cite{9140329}. 
Most importantly, by carefully tuning its passive reflecting beamforming, an IRS can achieve a considerable array performance gain.
With these advantages, IRSs have drawn so many interests recently with vast application prospects \cite{8982186,9301375}. The researchers employed IRSs to enhance the performance of communication systems, e.g., improving the data-rate \cite{9394419, 9090356,9110912,9680675}, reducing the transmit power \cite{8930608}, and enhancing energy efficiency \cite{9393607}, etc.   
Among them, one promising application was investigated to combine the IRS and the cell-free network to further enhance the network performance with the low hardware-cost and energy-consumption  \cite{9279253,9352948,9459505}. 
To be specific, the key idea in these works was taking advantage of IRSs to establish a reliable communication link between APs with UEs, so as to economically implement a virtual cell-free network. 
Accordingly, the efficient algorithms were provided to joint design the active transmitting beamforming and the passive reflecting beamforming of the AP and the IRS, respectively. 
Although such works were shown to be effective, they assumed that the high-accuracy full channel state information (CSI) knowledge was perfectly acquired at the APs, which was practically impossible due to the estimation and quantization errors \cite{9180053}. In fact, the efficient schemes have been extensively
studied to estimate CSI knowledge with satisfactory accuracy in IRS-assisted communication systems \cite{9418513,9241029,9130088,9354904}. However, for the IRS-assisted cell-free networks, the estimated CSI was still inaccurate due to that the dimensions of the channels are large, which makes CSI estimation errors are inevitable \cite{9366805}. 

To overcome the impact of the CSI uncertainty on the performance, a few works investigated the robust transmission design for the IRS-assisted communication system that take into account imperfect CSI knowledge. Most works in this area assumed a norm-bounded CSI error model, which is commonly employed when the CSI error is dominated by quantization errors \cite{5982443}. Under this error model, the authors in \cite{9468668,9266086,9133130} proposed the min-max (worst-case) constrained robust design schemes to fight against the CSI uncertainty. 
In general, the channel estimation error is modeled as a random variable, which follows a known probability distribution (e.g., Gaussian distribution) and is unbounded \cite{5982443,6515204,9293148}. 
As a result, a statistical CSI error model was applicable to characterize the practical imperfect CSI mainly due to the channel estimation errors \cite{9180053}.
The robustness under this error model were provided by solving the outage probability constrained problem or using the expected/averaged performance \cite{5982443,9293148,9180053}.
The robust designs for IRS-assisted MISO communication systems subjected to the rate outage probability constraints have been reported in \cite{9180053,9293148,9618858,9632613}.  
The authors in \cite{9117093} and \cite{9316283} proposed a pair of novel robust beamforming design schemes for the IRS-assisted MISO system, to minimize the mean squared error and to maximize the average sum-rate, respectively.
It would like to mention that the above works \cite{9117093,9316283} adopted the stochastic programming method to address the CSI uncertainties by using the expected or the averaged performance, though it does not ensure the robust performance for each individual realization \cite{5982443,6515204}.



Motivated by the discussions as mentioned above, in this paper, we investigate the robust joint design in the IRS-assisted cell-free MIMO communication network, where the statistical CSI error model of all channels are assumed. 
we adopt the stochastic programming method \cite{5982443,6515204} to cope with the CSI uncertainty by maximizing the expectation of the sum-rate, which guarantees robust performance over the average.
Accordingly, our goal is to maximize the average sum-rate by joint designing the active transmitting beamforming and the passive reflecting beamforming of the APs and the IRSs, respectively, while satisfying the power constraint at per AP and the constant modulus constraint at per PS of IRSs.

The main contributions of this paper are summarized as follows
\begin{itemize}
    \item We study the robust joint design of IRS-assisted cell-free MIMO networks in the presence of imperfect CSI knowledge. It is worth pointing out that the considered scenario is general, i.e., multiple APs and UEs with multiple transmitting and receiving antennas, multiple IRSs with multiple PSs, and all the channels are imperfect, and hence, is flexible to change to a specific case. Therefore, the proposed algorithm can act as a general framework to solve various sum-rate maximization problems of the IRS-assisted systems;  
    \item The formulated problem is challenging to solve due to the non-convex objective function, the intricately coupled variables, the constant modulus constraints, and the expectation operator about the uncertainty terms. To this end, we present an efficient robust joint design algorithm to handle it. Particularly, as a compromising approach, we first transform the original non-convex problem to an equivalent but tractable form by extending the fractional programming method, e.g., quadratic transform (QT) and  Lagrangian dual transform (LDT) into matrix-forms. Then, the transformed problem is  decomposed into two subproblems, i.e., the active transmit beamforming matrices optimization and the passive reflecting beamforming matrices optimization. Next, the block coordinate descent (BCD) -based method is adopted to solve the two subproblems, where both the locally optimal solutions of active and passive beamforming matrices in nearly closed-forms can be alternately optimized;
    \item Moreover, we further prove that the expectation value with respect to the CSI uncertainty surprisingly does not depend on phase-shifting of the passive reflecting beamforming of IRSs, but depends on the reflecting efficiency of IRSs, the total number of PSs at IRSs, and the active transmitting beamforming of APs. Consequently, the CSI uncertainty impacts the optimization of the active transmitting beamforming of APs, but surprisingly does not directly impact the 
   optimization of the passive reflecting beamforming of IRSs;
    \item Simulation results demonstrate that IRSs can significantly improve the rate compared with conventional cell-free networks, and confirm that the proposed robust joint design algorithm has the resilience against the imperfect CSI knowledge. 
\end{itemize}

\textit{Organizations:} The rest of this paper is organized as follows. In Section \ref{sec2}, we describe the IRS-assisted cell-free network model and formulate the average sum-rate maximization problem, while in Section \ref{sec3}, we propose an efficient robust joint design algorithm to solve it. Numerical simulation results are presented in Section \ref{sec4} and the paper is concluded in Section \ref{sec5}. 

\textit{Notations:} Vectors and matrices are presented by bold-face lower-case and upper-case letters, respectively. $\mathcal C\mathcal N\left(\mathbf 0, \mathbf I\right)$ denotes the circularly symmetric complex Gaussian (CSCG) distribution with zero mean and covariance matrix $\mathbf I_{\mathcal N}$, where $\mathbf I_{\mathcal N}$ denotes an $\mathcal N\times \mathcal N$ identity matrix. $\mathbb E\left\{ \cdot \right\}$ denotes the expectation operator. $\operatorname{Vec}\left\{ \cdot \right\}$ stacks all the columns of the argument into a single column vector and $\operatorname{Vecd}\left\{ \cdot \right\}$ forms a vector out
of the diagonal of its matrix argument. $\otimes$ and $ \odot $ denote the Kronecker and the Hadamard products, respectively. $\mathbf A^{\operatorname{H}}$ and $\operatorname{Tr}\left(\mathbf A\right)$ denote the Hermitian and trace operators of matrix $\mathbf A$, respectively. $\operatorname{Re}\left\{a\right\}$ is the real part of $a$.

\section{Network Model and Problem Formulation}{\label{sec2}}
\begin{figure}
    \centering
    \includegraphics[width=0.65\linewidth]{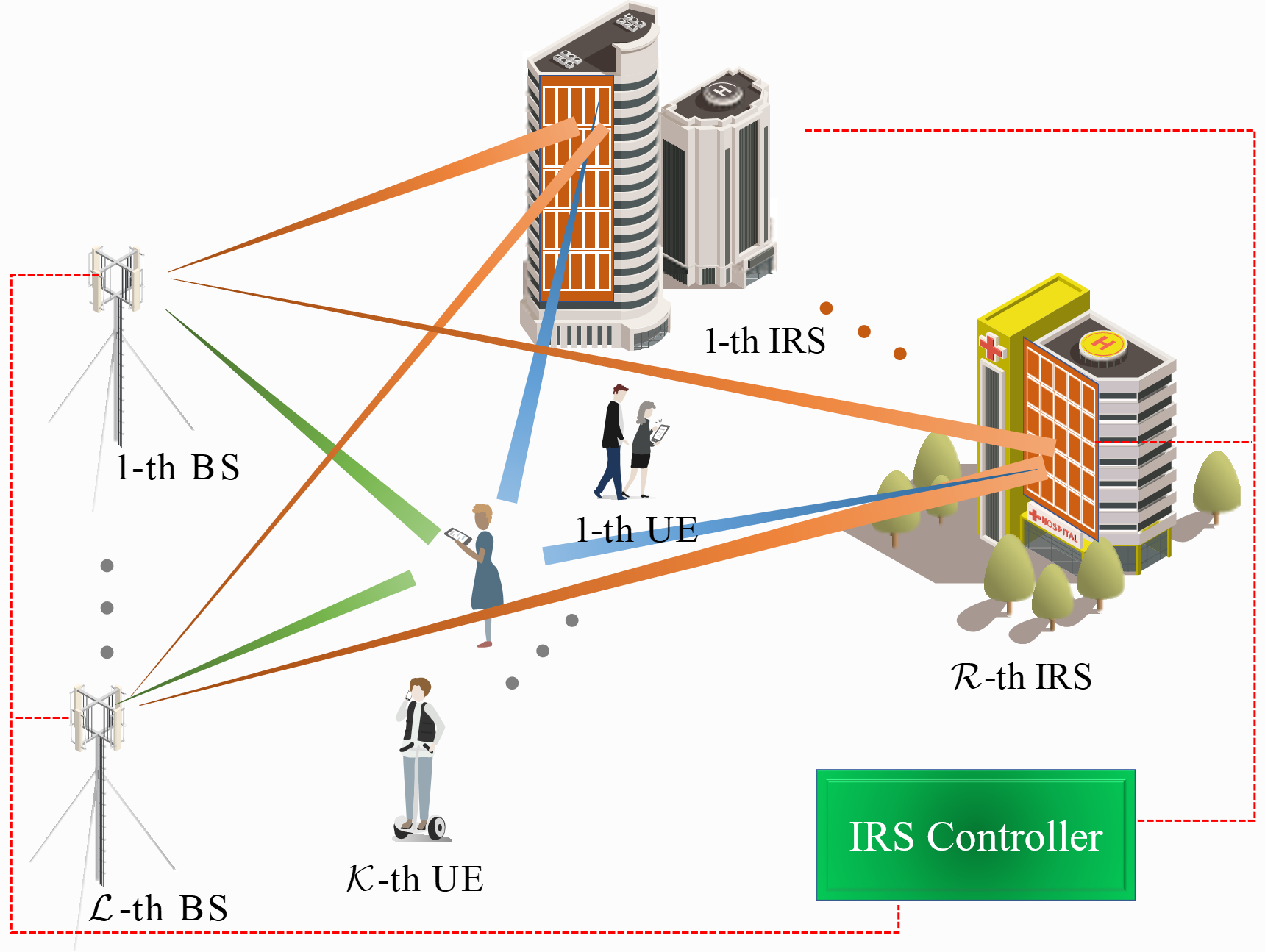}
    \caption{A system model schematic of IRSs-assisted cell-free MIMO networks.}
 \label{f1} 
\end{figure}
As shown in Fig. \ref{f1}, we describe a distributed multiple IRSs assisted cell-free MIMO downlink communication system \cite{9427474}, where $\mathcal L$ APs and $\mathcal R$ IRSs serve $\mathcal K$ UEs cooperatively. The IRSs are attached with a controller, which is responsible for tuning the phase-shifting of PSs according to the carefully optimized passive reflecting beamforming matrices. Moreover, each AP, IRS, and UE are equipped with $\mathcal M_B$ transmitting antennas, $\mathcal N$ PSs, and $\mathcal M_U$ receiving antennas, respectively. The signals reflected by the IRS two or more times are ignored \cite{8930608}.

\subsection{Channel Model}
We assume a narrow-band system and let ${\mathbf D_{l,k} \in {\mathbb{C}^{\mathcal M_B \times \mathcal M_U}}}$ denotes the direct channel from the $l$-th AP to the $k$-th UE, and ${\mathbf G_{r,k} \in \mathbb{C}^{\mathcal N \times \mathcal M_U}}$ is the channel from the $r$-th IRS to the $k$-th UE. The channel from the $l$-th AP to the $r$-th IRS is represented by ${\mathbf S_{l,r} \in {\mathbb{C}^{\mathcal N \times \mathcal M_B}}}$. 
We assume the small-scale fading of the direct channels (i.e., $\mathbf D_{l,k},\forall {l,k}$) and the IRS-related channels (i.e., $\mathbf G_{r,k},\forall {r,k}$ and $\mathbf S_{l,r},\forall {l,r}$) are Rayleigh fading and Rician fading \cite{9090356}, respectively.
Besides, we further assume that the APs and the UEs are equipped with uniform linear arrays (ULAs), and the IRSs are modeled as uniform planar arrays (UPAs) \cite{9133435,9423652,9279253}.

With the IRSs, the transmitted signal can be dynamically altered by its $\mathcal N$ PSs, where the diagonal passive reflecting beamforming matrices are denoted as
\begin{align}\label{equ1}
    \boldsymbol \Theta_r=\alpha\operatorname{diag}\left(e^{j{\phi_{r,1}}},e^{j{\phi_{r,2}}},\cdots,e^{j{\phi_{r,\mathcal N}}}\right)\in{\mathbb{C}^{\mathcal N \times \mathcal N}},\forall r \in \mathcal R,
\end{align}
where $\alpha \in \left[ {0,1} \right]$ denotes the reflecting efficiency of IRSs, which is used to measure the power loss caused by the signal absorption of the IRS, and $\phi_{r,n} \in \left[0,2\pi\right)$ is the phase-shifting of the $n$-th PS of the $r$-th IRS.

In general, to acquire the CSI knowledge with satisfactory accuracy is challenging to realize in IRS-assisted cell-free networks, and thus, to account for the estimation error of the acquired CSI knowledge, we consider CSI knowledge of both the direct and the IRS-related channels uncertain. Let $\mathbf H_{l,k} \in \mathbb{C}^{\mathcal M_B \times \mathcal M_U}$ denotes the equivalent actual channel spanning from the $l$-th AP to the $k$-th UE, which can be modeled as
\begin{align}\label{equ2}
    \mathbf { H}_{l,k}^{\operatorname{H}}=\mathbf {\hat D}_{l,k}^{\operatorname{H}}+\mathbf { \bar D}_{l,k}^{\operatorname{H}}+\sum_{r=1}^{\mathcal R}{\left(\mathbf {\hat G}_{r,k}^{\operatorname{H}}+\mathbf { \bar G}_{r,k}^{\operatorname{H}}\right)\boldsymbol{\Theta}_r\left(\mathbf {\hat S}_{l,r}+\mathbf { \bar S}_{l,r}\right)}, \forall l,k,
\end{align}
where $\mathbf  {\hat D}_{l,k}$, $\mathbf  {\hat G}_{r,k}$, and $\mathbf  {\hat S}_{l,r}$ are estimated CSI by using efficient estimation schemes for IRS-assisted systems \cite{9326394,9241029,9130088,9354904}, and $\mathbf {\bar D}_{l,k}$, $\mathbf {\bar G}_{r,k}$, and $\mathbf {\bar S}_{l,r}$ are the additive CSI estimation errors. 
As \cite{9180053,9293148,9618858,9632613,9117093,9316283 }, we assume the statistical CSI error model, which is applicable when the error is predominantly due to unavoidable inaccurate channel estimation in practical scenarios \cite{9293148}. Consequently, the CSI error matrices $\mathbf {\bar D}_{l,k}$, $\mathbf {\bar G}_{r,k}$, and $\mathbf {\bar S}_{l,r}$ are assumed to follow the Gaussian distribution \cite{9180053,5982443} with zero mean and 
\begin{subequations}
\begin{align}\label{equ3}
    \mathbb{E}\left\{\operatorname{Vec}\left(\mathbf {\bar D}_{l,k}\right)\operatorname{Vec}\left(\mathbf {\bar D}_{l,k}\right)^{\operatorname{H}}\right\}&=\delta_{\mathbf D_{l,k}}^2\mathbf I_{\mathcal M_B\mathcal M_U}, \forall l,k,\\
    \mathbb{E}\left\{\operatorname{Vec}\left(\mathbf {\bar G}_{r,k}\right)\operatorname{Vec}\left(\mathbf {\bar G}_{r,k}\right)^{\operatorname{H}}\right\}&=\delta_{\mathbf G_{r,k}}^2\mathbf I_{\mathcal N\mathcal M_U}, \forall r,k,\\
    \mathbb{E}\left\{\operatorname{Vec}\left(\mathbf {\bar S}_{l,r}\right)\operatorname{Vec}\left(\mathbf {\bar S}_{l,r}\right)^{\operatorname{H}}\right\}&=\delta_{\mathbf S_{l,r}}^2\mathbf I_{\mathcal N\mathcal M_B}, \forall l,r,
\end{align}
\end{subequations}
where $\delta_{\mathbf D_{l,k}}^2=\kappa_{\mathbf D}^2\left\| \operatorname{Vec}\left(\mathbf {\hat D}_{l,k}\right)\right\|_{2}^2$, $\delta_{\mathbf G_{r,k}}^2=\kappa_{\mathbf G}^2\left\|\operatorname{Vec}\left(\mathbf {\hat G}_{r,k}\right)\right\|_{2}^2$, and $\delta_{\mathbf S_{l,r}}^2=\kappa_{\mathbf S}^2\left\|\operatorname{Vec}\left(\mathbf {\hat S}_{l,r}\right)\right\|_{2}^2$, with $\kappa_{ \mathbf D /\mathbf G/\mathbf S }^2\in \left[0,1\right)$ are normalized CSI errors, which are used to measure the relative amount of the CSI uncertainties.
Further, all the CSI error covariances are assumed to be known at the APs \cite{5982443}. 

By defining $\boldsymbol\Theta= \operatorname{blkdiag}\left\{\boldsymbol\Theta_1,\boldsymbol\Theta_2,\cdots,\boldsymbol\Theta_{\mathcal R}\right\}\in \mathbb C^{\mathcal R\mathcal N\times \mathcal R\mathcal N}$, $\mathbf G_{k}=\left[\mathbf G_{1,k}^{\operatorname{T}},\mathbf G_{2,k}^{\operatorname{T}},\cdots,\mathbf G_{\mathcal R,k}^{\operatorname{T}}\right]^{\operatorname{T}}\in \mathbb C^{\mathcal R\mathcal N\times \mathcal M_U}$, $\mathbf S_{l}=\left[\mathbf S_{l,1}^{\operatorname{T}},\mathbf S_{l,2}^{\operatorname{T}},\cdots,\mathbf S_{l,\mathcal R}^{\operatorname{T}}\right]^{\operatorname{T}}\in \mathbb C^{\mathcal R\mathcal N\times \mathcal M_B}$, we have 
\begin{align}\label{equ4}
    \mathbf { H}_{l,k}^{\operatorname{H}}\triangleq\mathbf {\hat H}_{l,k}^{\operatorname{H}}+\mathbf {\bar H}_{l,k}^{\operatorname{H}}, \forall l,k,
\end{align}
where $\mathbf {\hat H}_{l,k}^{\operatorname{H}}\triangleq\mathbf {\hat D}_{l,k}^{\operatorname{H}}+\mathbf {\hat G}_{k}^{\operatorname{H}}\boldsymbol{\Theta}\mathbf {\hat S}_{l}^{\operatorname{H}}$ and $\mathbf {\bar  H}_{l,k}^{\operatorname{H}}\triangleq\mathbf { \bar D}_{l,k}^{\operatorname{H}}+\mathbf {\bar  G}_{k}^{\operatorname{H}}\boldsymbol{\Theta}\mathbf {\hat S}_{l}^{\operatorname{H}}+\mathbf {\hat G}_{k}^{\operatorname{H}}\boldsymbol{\Theta}\mathbf {\bar S}_{l}^{\operatorname{H}}+\mathbf { \bar G}_{k}^{\operatorname{H}}\boldsymbol{\Theta}\mathbf {\bar S}_{l}^{\operatorname{H}}$.

\subsection{Network Model}

The signal transmitted from the $\mathcal L$ APs can be mathematically expressed as
\begin{align}\label{equ5}
\mathbf{x} = \sum_{l=1}^{\mathcal L}{\sum_{k = 1}^{\mathcal K} \mathbf{W}_{l,k}\mathbf{s}_{l,k}}, \forall l \in \mathcal L,\forall k \in \mathcal K.
\end{align}
where ${\mathbf s_{l,k}} \in {\mathbb{C}^{d \times 1}}$ denotes $d$ desired data streams from the $l$-th AP intend for the $k$-th UE and satisfies ${\mathbf s_{l,k}} \sim\mathcal C\mathcal N\left(\mathbf 0, \mathbf I_{d}\right)$. 
$\mathbf W _{l,k} \in {\mathbb{C}^{\mathcal M_B \times d}}$ is the corresponding active transmitting beamforming matrix from the $l$-th AP to the $k$-th UE, the transmit power at each AP is $\sum_{k=1}^{\mathcal K}{\mathbb E\left\{\mathbf s_{l,k}^{\operatorname{H}}\mathbf W_{l,k}^{\operatorname{H}}\mathbf W_{l,k}\mathbf s_{l,k}\right\}}=\sum_{k=1}^{\mathcal K}{\left\|\mathbf W_{l,k}\right\|_{\operatorname{F}}^2}$.

The received signal at the $k$-th UE can be expressed by
\begin{align}\label{equ6}
{{\mathbf y_{k}}} = \sum_{l=1}^{\mathcal L}{\mathbf {\hat H} _{l,k}^{\operatorname{H}} \mathbf W _{l,k} \mathbf s_{l,k}}+\sum_{l=1}^{\mathcal L}\sum_{i=1, i\ne k}^{\mathcal K}{\mathbf {\hat H} _{l,k}^{\operatorname{H}} \mathbf W _{l,i} \mathbf s_{l,i}}+\sum_{l=1}^{\mathcal L}\sum_{i=1}^{\mathcal K}{\mathbf {\bar H} _{l,k}^{\operatorname{H}} \mathbf W _{l,i} \mathbf s_{l,i}}+\mathbf n_k,
\end{align}
where $\mathbf n_{k} \sim  \mathcal{C}\mathcal{N}\left( {\mathbf 0,{\sigma_k ^2\mathbf {I}_{\mathcal M_U}}} \right)$ is the background additive white Gaussian noise (AWGN) vector at the $k$-th UE.

Based on the aforementioned discussions, the signal to interference plus noise ratio (SINR) matrix at the $k$-th UE can be formulated as 
\begin{align}\label{equ7}
    {\boldsymbol{\Gamma} _{k}} =\sum_{l=1}^{\mathcal L}{ \mathbf {\hat H}_{l,k}^{\operatorname{H}}{\mathbf W_{l,k}}\mathbf V_k^{\operatorname{-1}}\mathbf W_{l,k}^{\operatorname{H}}{\mathbf {\hat H}_{l,k}}},\forall k \in \mathcal K,
\end{align}
where $\mathbf V_k$ denotes the covariance of the effective interference plus noise and is calculated as
\begin{align}\label{equ8}
  \mathbf V_{k} = &\sum_{l = 1}^{\mathcal L}\sum_{i=1, i\ne k}^{\mathcal K}{\mathbf {\hat H}_{l,k}^{\operatorname{H}}{\mathbf W_{l,i}}\mathbf W_{l,i}^{\operatorname{H}}{\mathbf {\hat H}_{l,k}}}  + \sum_{l = 1}^{\mathcal L}\sum_{i=1}^{\mathcal K}{\mathbf {\bar H}_{l,k}^{\operatorname{H}}{\mathbf W_{l,i}}\mathbf W_{l,i}^{\operatorname{H}}{\mathbf {\bar H}_{l,k}}}+ {\sigma_k ^2}{\mathbf {I}_{{\mathcal M_U}}}. 
\end{align}

We adopt a stochastic programming method to address the CSI uncertainties by maximizing the expectation of the sum-rate that depend on the CSI error. Such a stochastic programming method guarantees robust performance over the average, though it does not ensure robust performance for each individual realization \cite{5982443}. By collecting all active transmit beamforming matrices at $\mathcal L$ APs as $\mathbf {\tilde W}=\left\{\mathbf W_{l,k},\forall {l,k}\right\}$ and passive reflecting beamforming matrices at $\mathcal R$ IRSs as $\boldsymbol{\tilde\Theta}=\left\{\boldsymbol\Theta_r,\forall r\right\}$, the average sum-rate of the IRS-assisted cell-free network is given by
\begin{align}
    \mathcal{R}\left( {\mathbf {\tilde W},\boldsymbol{\tilde\Theta}} \right) =\mathbb E\left\{ {\sum_{k = 1}^{\mathcal K} {\log \left| {\mathbf I_{{\mathcal M_U}} + {\boldsymbol\Gamma _{k}}} \right|} }\right\},
\end{align}
where the expectation is taken over all uncertain terms, i.e., the unknown CSI estimation error matrices $\mathbf {\bar D}_{l,k}$, $\mathbf {\bar G}_{r,k}$, and $\mathbf {\bar S}_{l,r}, \forall l,r,k$, which is main challenge to solve this problem. 

\subsection{Problem Formulation}

In this paper, our objective is maximizing the average sum-rate through joint optimizing the robust active transmit beamforming $\mathbf {\tilde W}$ and the robust passive reflecting beamforming $\boldsymbol {\tilde \Theta}$ with the imperfect CSI knowledge, while satisfying the power constraint of per AP and the constant modulus constraint of per PS. Accordingly, the average sum-rate of the IRS-assisted cell-free network maximization problem can be mathematically formulated as
\begin{subequations}\label{equ10}
\begin{align}
\mathcal P_1:\;{ \mathop {\max }_{\mathbf {\tilde W} ,\boldsymbol{\tilde\Theta} } } \quad & {{\mathcal{R}}\left( {\mathbf {\tilde W},\boldsymbol{\tilde\Theta}  } \right)} \label{equ10a}\\
{\operatorname{s.t.} \quad} & {\sum_{k = 1}^{\mathcal K} {{{\left\| {{\mathbf W _{l,k}}} \right\|}_{\operatorname{F}} ^2}}  \le {P_{l}^{\max}},\forall\, l \in \mathcal  L}\label{equ10b},\\
&  \left  |\boldsymbol\Theta_{n,n} \right |={\alpha},\forall n \in\mathcal R\mathcal N, \label{equ10c}
\end{align} 
\end{subequations}
where the constraint \eqref{equ10b} limits the maximum transmit power of each AP and \eqref{equ10c} represents the constant modulus constraint of each PS at the IRSs. Due to the optimizing variables are intricately coupled in the matrix-ratio terms, problem $\mathcal P_1$ is non-convex and arduous to tackle by employing the existing methods directly. Besides, the expectation operator with respect to uncertain terms, i.e., CSI estimation errors, makes the formulated problem becomes rather challenging. Therefore, in the following section, we provide a robust joint design algorithm to solve the above non-convex problem.

\section{Proposed Robust Joint Design Algorithm}{\label{sec3}}
In this section, we describe an efficient algorithm to robust joint design the active transmitting beamforming matrices of APs and the passive reflecting beamforming matrices of IRSs in the presence of imperfect CSI knowledge.
\subsection{Overview of the Proposed Algorithm}
Under the statistical CSI error model, we adopt a stochastic programming method to guarantee the robust performance and formulate problem $\mathcal P_1$. Particularly, we propose the robust joint design algorithm to tackle this problem, and locally optimal solutions of the active transmitting beamforming of APs and the passive reflecting beamforming of IRSs can be obtained.

First, as a compromising approach, we consider to transform intractable problem $\mathcal P_1$ into an equivalent form by extending the well-known fractional programming (FP) methods, e.g., quadratic transform (QT) 
\cite[Corollary 1]{shen2018fractional2} and Lagrangian dual transform (LDT) \cite[Theorem 4]{shen2018fractional2} into matrix-forms. Accordingly, we have the following proposition:

\begin{proposition}\label{p2}
By introducing a pair of auxiliary matrices $\mathbf {\tilde U}=\left\{\mathbf U_k\in \mathbb{C}^{d \times d}, \forall k\right\}$  and $\mathbf {\tilde Y}=\left\{\mathbf Y_k\in \mathbb{C}^{\mathcal M_U \times d}, \forall k\right\}$, the formulated problem $\mathcal P_1$ can be equivalently transformed as
\begin{subequations}
\begin{align}\label{equ11}
   \mathcal P_2:\; \mathop{\max}_{\mathbf {\tilde W}, \boldsymbol{\tilde\Theta}, \mathbf {\tilde U}, \mathbf {\tilde Y}} \quad & f_1\left(\mathbf {\tilde W}, \boldsymbol{\tilde\Theta}, \mathbf {\tilde U}, \mathbf {\tilde Y}\right)\\
    {\operatorname{s.t.}}\quad &\eqref{equ10b} \,\operatorname{and}\,\eqref{equ10c},
\end{align}
\end{subequations}
where the objective function is given by
\begin{align}\label{equ12}
    &f_1\left(\mathbf {\tilde W}, \boldsymbol{\tilde\Theta}, \mathbf {\tilde U}, \mathbf {\tilde Y}\right)=\sum_{k=1}^{\mathcal K}{\log\left|\mathbf {\bar U}_k\right|}-\sum_{k=1}^{\mathcal K}{\operatorname{Tr}\left(\mathbf U_k\right)}+\sum_{k=1}^{\mathcal K}{\operatorname{Tr}\left(\mathbf {\bar U}_k \mathbf Y_k^{\operatorname{H}}\sum_{l=1}^{\mathcal L}{\mathbf {\hat H}_{l,k}^{\operatorname{H}}\mathbf W_{l,k}}\right)}\notag\\
    &+\sum_{k=1}^{\mathcal K}{\operatorname{Tr}\left(\mathbf {\bar U}_k \sum_{l=1}^{\mathcal L}{\mathbf W_{l,k}^{\operatorname{H}}\mathbf {\hat H}_{l,k}}\mathbf Y_k\right)}-\sum_{k=1}^{\mathcal K}{\operatorname{Tr}\left({\mathbf{\bar U}_k}\mathbf Y_k^{\operatorname{H}}\sum_{l=1}^{\mathcal L}\sum_{i=1}^{\mathcal K}{\mathbf {\hat H}_{l,k}^{\operatorname{H}}\mathbf W_{l,i}\mathbf W_{l,i}^{\operatorname{H}}\mathbf {\hat H}_{l,k}}\mathbf Y_{k}\right)}\notag\\
    &-\mathbb E\left\{\sum_{k=1}^{\mathcal K}{\operatorname{Tr}\left({\mathbf{\bar U}_k}\mathbf Y_k^{\operatorname{H}}\sum_{l=1}^{\mathcal L}\sum_{i=1}^{\mathcal K}{\mathbf {\bar H}_{l,k}^{\operatorname{H}}\mathbf W_{l,i}\mathbf W_{l,i}^{\operatorname{H}}\mathbf {\bar H}_{l,k}}\mathbf Y_{k}\right)}\right\}
    -\sum_{k=1}^{\mathcal K}{\operatorname{Tr}\left(\sigma_k^2{\mathbf{\bar U}_k}{\mathbf Y_k^{\operatorname{H}}}\mathbf Y_{k}\right)}.
\end{align}
where $\mathbf {\bar U}_k \triangleq {\mathbf {U} _{k}}+\mathbf I_{\mathcal M_u},\forall k \in {\mathcal K}$.
\end{proposition}
\begin{proof}
This proposition is extending QT and FPT methods from vector-forms into matrix-forms, and the proof can follows the result in \cite{shen2018fractional,shen2018fractional2}, and thus, is
omitted here for brevity.
\end{proof}

Although we have introduced two additional optimization variables, problem $\mathcal P_2$ has been significantly simplified.
However, problem $\mathcal P_2$ is still hard to solve due to the presence of the expectation operation with respect to the CSI uncertainties $\mathbf {\bar D}_{l,k}, \forall l,k$, $\mathbf {\bar G}_{r,k}, \forall r,k$, and $\mathbf {\bar S}_{l,r}, \forall l,r$. To this end, we first investigate the expectation term and have the following proposition
\begin{proposition}\label{proepec}
The expectation term $\mathbb E\left\{\sum_{k=1}^{\mathcal K}{\operatorname{Tr}\left({\mathbf{\bar U}_k}\mathbf Y_k^{\operatorname{H}}\sum_{l=1}^{\mathcal L}\sum_{i=1}^{\mathcal K}{\mathbf {\bar H}_{l,k}^{\operatorname{H}}\mathbf W_{l,i}\mathbf W_{l,i}^{\operatorname{H}}\mathbf {\bar H}_{l,k}}\mathbf Y_{k}\right)}\right\}$ in \eqref{equ12} can be expressed by
\begin{align}\label{equ13}
&\mathbb E\left\{\sum_{k=1}^{\mathcal K}{\operatorname{Tr}\left({\mathbf{\bar U}_k}\mathbf Y_k^{\operatorname{H}}\sum_{l=1}^{\mathcal L}\sum_{i=1}^{\mathcal K}{\mathbf {\bar H}_{l,k}^{\operatorname{H}}\mathbf W_{l,i}\mathbf W_{l,i}^{\operatorname{H}}\mathbf {\bar H}_{l,k}}\mathbf Y_{k}\right)}\right\}\notag\\
    =&\sum_{k=1}^{\mathcal K}{\sum_{l=1}^{\mathcal L}{\sum_{i=1}^{\mathcal K}{\delta_{\mathbf D_{l,k}}^2\operatorname{Tr}\left(\mathbf Y_k{\mathbf{\bar U}_k}\mathbf Y_k^{\operatorname{H}}\right)\operatorname{Tr}\left(\mathbf W_{l,i}\mathbf W_{l,i}^{\operatorname{H}}\right)}}}\notag\\
    &+\sum_{k=1}^{\mathcal K}{\sum_{l=1}^{\mathcal L}{\sum_{i=1}^{\mathcal K}{\alpha^2\delta_{\mathbf G_{k}}^2\operatorname{Tr}\left(\mathbf Y_k{\mathbf{\bar U}_k}\mathbf Y_k^{\operatorname{H}}\right)\operatorname{Tr}\left(\mathbf {\hat S}_l\mathbf W_{l,i}\mathbf W_{l,i}^{\operatorname{H}}\mathbf {\hat S}_l^{\operatorname{H}}\right)}}}\notag\\
    &+\sum_{k=1}^{\mathcal K}{\sum_{l=1}^{\mathcal L}{\sum_{i=1}^{\mathcal K}{\alpha^2\delta_{\mathbf S_{l}}^2\operatorname{Tr}\left(\mathbf {\hat G}_k\mathbf Y_k{\mathbf{\bar U}_k}\mathbf Y_k^{\operatorname{H}}\mathbf {\hat G}_k^{\operatorname{H}}\right)\operatorname{Tr}\left(\mathbf W_{l,i}\mathbf W_{l,i}^{\operatorname{H}}\right)}}}\notag\\
    &+\sum_{k=1}^{\mathcal K}{\sum_{l=1}^{\mathcal L}{\sum_{i=1}^{\mathcal K}{\mathcal R\mathcal N\alpha^2\delta_{\mathbf S_{l}}^2\delta_{\mathbf G_{k}}^2\operatorname{Tr}\left(\mathbf Y_k{\mathbf{\bar U}_k}\mathbf Y_k^{\operatorname{H}}\right)\operatorname{Tr}\left(\mathbf W_{l,i}\mathbf W_{l,i}^{\operatorname{H}}\right)}}}.
\end{align}
\end{proposition}
\begin{proof}
The complete proof is given in Appendix \ref{appexp}.
\end{proof}

The proposition verify that the expectation term in \eqref{equ12} surprisingly does not depend on phase-shifting of the passive reflecting beamforming matrices of IRSs, i.e., $\boldsymbol{\widetilde\Theta}$, while  depends on the reflecting efficiency of IRSs, i.e., $\alpha$, the total number of PSs at IRSs, i.e., $\mathcal R\times\mathcal N$, the introduced auxiliary matrices, i.e., $\mathbf {\widetilde U}$ and $\mathbf {\widetilde Y}$, and the active transmitting beamforming matrices of APs, i.e., $\mathbf {\widetilde W}$. Consequently, the CSI uncertainty impacts the optimizations of the introduced auxiliary matrices and the active transmitting beamforming matrices of APs, but surprisingly does not directly impact the optimization of the passive reflecting beamforming of IRSs.

Note that the objective function in problem $\mathcal P_2$ is convex with respect to any one of the four sets of variables $\mathbf {\tilde U}$, $\mathbf {\tilde Y}$, $\mathbf {\tilde W}$, and $\boldsymbol{\tilde \Theta}$ when the other three being fixed. Thus, inspired by this fact, in the rest of this section, we decompose problem $\mathcal P_2$ into several subproblems and solve them by employing the BCD algorithm \cite{doi:10.1137/120891009} to obtain locally optimal solutions iteratively until converge. 
The key steps for designing the IRS-assisted cell-free network are drawing as a flow chart in Fig.\ref{fc}. The solutions after the $t$-th iteration are denoted by $\left(\cdot\right)^{\left(\operatorname{t+1}\right)}$.

\begin{figure}
    \centering
 \includegraphics[width=0.5\linewidth]{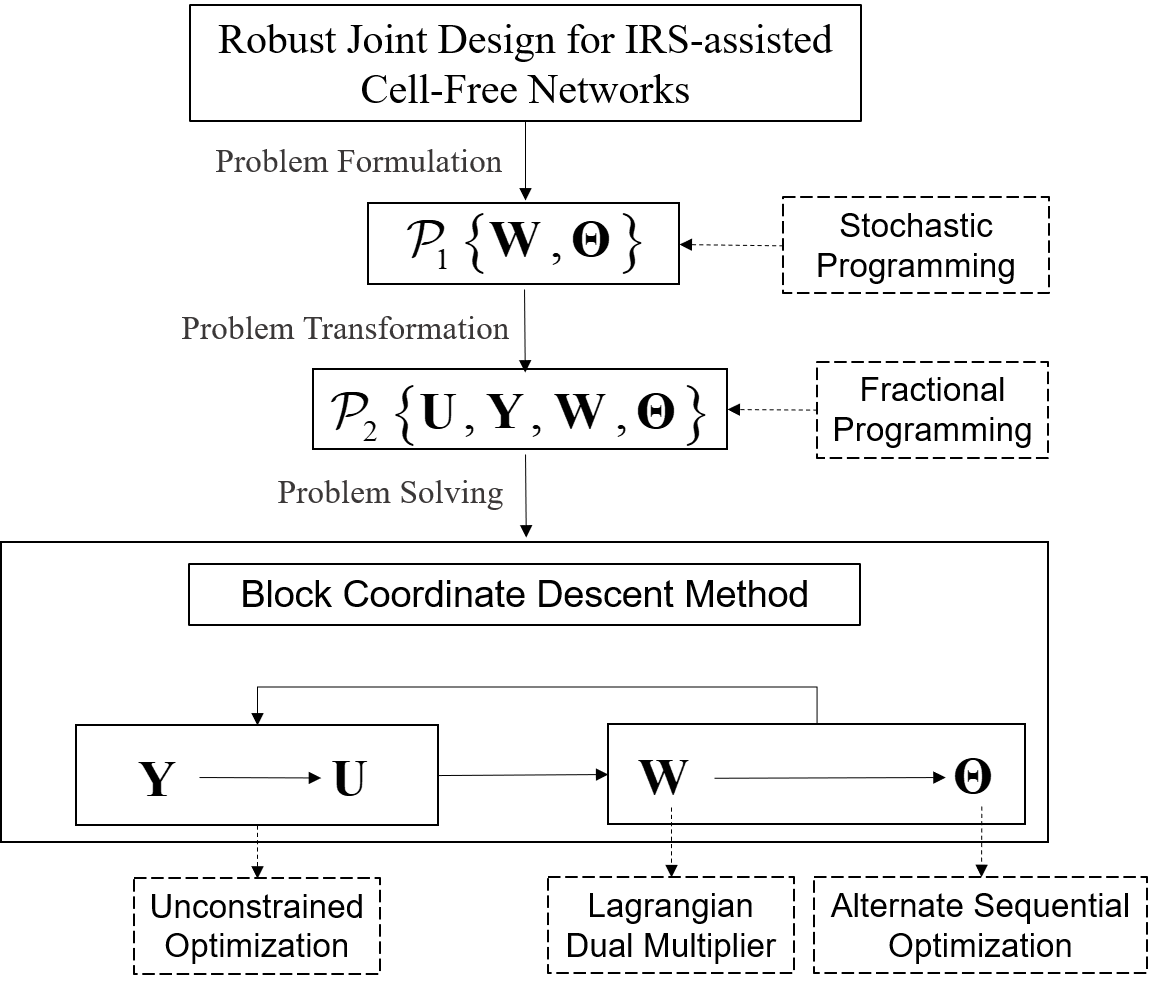}
    \caption{A flow chart of the proposed robust joint design algorithm.}
    \label{fc}
\end{figure}

\subsection{Optimization of Auxiliary Matrices}
First, with fixed variables of $\mathbf {\tilde W}^{\left(\operatorname{t}\right)}$ and $\boldsymbol{\tilde \Theta}^{\left(\operatorname{t}\right)}$, the optimal auxiliary variables $\mathbf {\tilde Y}^{\left(\operatorname{t+1}\right)}$ and $\mathbf {\tilde U}^{\left(\operatorname{t+1}\right)}$ can be determined. Particularly, note that the auxiliary variables $\mathbf {\tilde Y}$ and $\mathbf {\tilde U}$ only appear in the objective function and do not exist in any constraint sets, which implies that the calculation subproblems of the two auxiliary matrices constitute a pair of unconstrained optimization problems. As a result, the solutions can be obtained by setting the partial derivatives of $f_1\left(\mathbf {\tilde W}, \boldsymbol{\tilde\Theta}, \mathbf {\tilde U}, \mathbf {\tilde Y}\right)$ with respect to $\mathbf {\tilde Y}$ and $\mathbf {\tilde U}$ to be zeros, respectively.

After the matrix manipulations, the closed-form solution of $\mathbf {\tilde Y}$ can be expressed by
\begin{align}\label{equ14}
    \mathbf Y_k^{\left(\operatorname{t+1}\right)}=\mathbf {{\widetilde V}}_k^{\operatorname{-1}}\sum_{l=1}^{\mathcal L}{\mathbf {\hat H}_{l,k}^{\operatorname{H}}\mathbf W_{l,k}}, \forall k.
\end{align}
where 
\begin{align}
    \mathbf {{\widetilde V}}_k=&\sum_{l=1}^{\mathcal L}\sum_{i=1}^{\mathcal K}{\mathbf {\hat H}_{l,k}^{\operatorname{H}}\mathbf W_{l,i}\mathbf W_{l,i}^{\operatorname{H}}\mathbf {\hat H}_{l,k}}\\\notag
&+\sum_{l=1}^{\mathcal L}\sum_{i=1}^{\mathcal K}{\left(\delta_{\mathbf D_{l,k}}^2\mathbf I_{\mathcal M_U}+\alpha^2\delta_{\mathbf S_l}^2\mathbf {\hat G}_k^{\operatorname{H}}\mathbf {\hat G}_k+\alpha^2\mathcal R\mathcal N\delta_{\mathbf G_k}^2\delta_{\mathbf S_l}^2\mathbf I_{\mathcal M_U}\right)\operatorname{Tr}\left({\mathbf W_{l,i}\mathbf W_{l,i}^{\operatorname{H}}} \right)}\\\notag
&+\sum_{l=1}^{\mathcal L}\sum_{i=1}^{\mathcal K}{\alpha^2\delta_{\mathbf G_k}^2\mathbf I_{\mathcal M_U}\operatorname{Tr}\left({\mathbf {\hat S}_l\mathbf W_{l,i}\mathbf W_{l,i}^{\operatorname{H}}\mathbf {\hat S}_l^{\operatorname{H}}}\right)}+\sigma_k^2\mathbf I_{\mathcal M_U}.
\end{align}
The optimal solution of the auxiliary matrices $\mathbf {\tilde Y}$ determined by \eqref{equ14} is known as minimum
mean-square error (MMSE) decoding filter \cite{5756489}.

Then, by substituting the obtained optimal solution of $\mathbf {\tilde Y}^{\left(\operatorname{t+1}\right)}$ into $f_1\left(\mathbf {\tilde W}, \boldsymbol{\tilde\Theta}, \mathbf {\tilde U}, \mathbf {\tilde Y}\right)$, the closed-form solution of $\mathbf {\tilde U}^{\left(\operatorname{t+1}\right)}$ can be determined as follows
\begin{align}\label{equ15}
    \mathbf U_k^{\left(\operatorname{t+1}\right)}= \sum_{l=1}^{\mathcal L}{\mathbf W_{l,k}^{\operatorname{H}}\mathbf {\hat H}_{l,k}\mathbf {\tilde V}_k^{\operatorname{-1}}\mathbf {\hat H}_{l,k}^{\operatorname{H}}\mathbf W_{l,k}}, \forall k.
\end{align}
where $\mathbf {\tilde V}_k=\mathbf {\widetilde V}_k-\sum_{l=1}^{\mathcal L}{{\mathbf {\hat H}_{l,k}^{\operatorname{H}}\mathbf W_{l,k}\mathbf W_{l,k}^{\operatorname{H}}\mathbf {\hat H}_{l,k}}}$.

\subsection{Optimization of Active Transmitting Beamforming}
In this subsection, we study the optimization of the active transmitting beamforming $\mathbf {\tilde W}$, while fixing $\mathbf {\tilde Y} ^{\operatorname{\left(t+1\right)}}$, $\mathbf {\tilde U} ^{\operatorname{\left(t+1\right)}}$, and $\boldsymbol {\tilde \Theta} ^{\operatorname{\left(t\right)}}$. By omitting the irrelevant constant terms with respect to $\mathbf {\tilde W}$, i.e., $\sum_{k=1}^{\mathcal K}{\log\left|\mathbf {\bar U}_k\right|}$, $\sum_{k=1}^{\mathcal K}{\operatorname{Tr}\left(\mathbf U_k\right)}$, and $\sum_{k=1}^{\mathcal K}{\operatorname{Tr}\left(\sigma_k^2\mathbf {\bar U}_k\mathbf Y_k^{\operatorname{H}}\mathbf Y_k\right)}$, which have no impact on updating of $\mathbf {\tilde W}$, the subproblem of optimizing the active transmitting beamforming can be simplified by
\begin{subequations}
\begin{align}
  \mathcal P_3:\;  \mathbf {\tilde W}^{\left(\operatorname{t+1}\right)} \triangleq \arg\mathop{\min} _\mathbf {\tilde W} \quad &\sum_{l=1}^{\mathcal L}{\sum_{k=1}^{\mathcal K}{\operatorname{Tr}\left(\mathbf W_{l,k}^{\operatorname{H}}\mathbf A_{l}\mathbf W_{l,k}\right)}}\notag\\
    &-\sum_{l=1}^{\mathcal L}{\sum_{k=1}^{\mathcal K}{\operatorname{Tr}\left(\mathbf {\bar U}_k \mathbf Y_k^{\operatorname{H}}{\mathbf {\hat H}_{l,k}^{\operatorname{H}}\mathbf W_{l,k}}\right)}}\notag\\
    &-\sum_{l=1}^{\mathcal L}{\sum_{k=1}^{\mathcal K}{\operatorname{Tr}\left(\mathbf {\bar U}_k \mathbf W_{l,k}^{\operatorname{H}}\mathbf {\hat H}_{l,k}\mathbf Y_k\right)}},\label{equ16a}\\
    \quad \operatorname{s.t.} \quad& \sum_{k=1}^{\mathcal K}{\operatorname{Tr}\left(\mathbf W_{l,k}^{\operatorname{H}}\mathbf W_{l,k}\right)}\le\operatorname{P}_{l}^{\max},\forall l\in \mathcal L,\label{equ16b}
\end{align}
\end{subequations}
where 
\begin{align}
\mathbf A_l=&\sum_{k=1}^{\mathcal K}{\mathbf {\hat H}_{l,k}\mathbf Y_k\mathbf {\tilde U}_k\mathbf Y_k^{\operatorname{H}}\mathbf {\hat H}_{l,k}^{\operatorname{H}}}
+\sum_{k=1}^{\mathcal K}\left(\delta_{\mathbf D_{l,k}}^2
+\alpha^2\mathcal R\mathcal N\delta_{\mathbf G_k}^2\delta_{\mathbf S_l}^2\right){\operatorname{Tr}\left(\mathbf Y_k{\mathbf{\bar U}_k}\mathbf Y_k^{\operatorname{H}}\right)}\mathbf I_{\mathcal M_B}\notag\\
&+\sum_{k=1}^{\mathcal K}{\alpha^2\delta_{\mathbf G_k}^2\operatorname{Tr}\left(\mathbf Y_k{\mathbf{\bar U}_k}\mathbf Y_k^{\operatorname{H}}\right)}\mathbf {\hat S}_l^{\operatorname{H}}\mathbf {\hat S}_l
+\sum_{k=1}^{\mathcal K}{\alpha^2\delta_{\mathbf S}^2\operatorname{Tr}\left(\mathbf {\hat G}_k\mathbf Y_k{\mathbf{\bar U}_k}\mathbf Y_k^{\operatorname{H}}\mathbf {\hat G}_k^{\operatorname{H}}\right)}\mathbf I_{\mathcal M_B}.
\end{align}

It can be verified that the objective function in \eqref{equ16a} and the constraint in \eqref{equ16b} are both convex with respect to in $\mathbf {\tilde W}$, and thus, problem $\mathcal P_3$ constitutes a convex optimization problem, which can be efficiently solved by employing the generic convex solvers, e.g., CVX \cite{cvx}. Instead of relying on the generic solver with high computational complexity, we provide a locally optimal solution in nearly closed-form of $\mathbf {\tilde W}$ by employing Lagrangian dual multiplier method \cite{boyd2004convex}. 

Since problem $\mathcal P_3$ satisfies the Slater's condition and the dual gap is guaranteed to be zero, thus, the optimal solution can be obtained by solving its dual problem instead of its original one \cite{6336836}. The Karush–Kuhn–Tucker (KKT) conditions of problem $\mathcal P_3$ with respect to $\mathbf {\tilde W}$ are represented as
\begin{subequations}
\begin{align}
    \nabla_{\mathbf W_{l,k}}{\mathcal L\left(\mathbf {\tilde W},\tilde\lambda\right)}= 2\frac{\partial \mathcal L\left(\mathbf {\tilde W},\tilde\lambda\right)}{\partial {\mathbf W_{l,k}^{\ast}}}&=\boldsymbol{0}, \forall l,k,\label{equ18a}\\
    \lambda_l\left(\sum_{k=1}^{\mathcal K}{\operatorname{Tr}\left(\mathbf W_{l,k}^{\operatorname{H}}\left(\lambda_l\right)\mathbf W_{l,k}\left(\lambda_l\right)\right)}-\operatorname{P}_{l}^{\max}\right)&=0, \forall l,\label{equ18b}
\end{align}
\end{subequations}
where $\mathcal L\left(\mathbf {\tilde W},\tilde\lambda\right)$ is the Lagrangian dual function of problem $\mathcal P_3$ and can be expressed by
\begin{align}
    \mathcal L\left(\mathbf {\tilde W},\tilde\lambda\right)&=\sum_{l=1}^{\mathcal L}{\sum_{k=1}^{\mathcal K}{\operatorname{Tr}\left(\mathbf W_{l,k}^{\operatorname{H}}\left(\mathbf A_{l}+\lambda_l\mathbf I_{\mathcal M_B}\right)\mathbf W_{l,k}\right)}}-\sum_{l=1}^{\mathcal L}{\sum_{k=1}^{\mathcal K}{\operatorname{Tr}\left(\mathbf {\bar U}_k \mathbf Y_k^{\operatorname{H}}{\mathbf {\hat H}_{l,k}^{\operatorname{H}}\mathbf W_{l,k}}\right)}}\notag\\
    &-\sum_{l=1}^{\mathcal L}{\sum_{k=1}^{\mathcal K}{\operatorname{Tr}\left(\mathbf {\bar U}_k \mathbf W_{l,k}^{\operatorname{H}}\mathbf {\hat H}_{l,k}\mathbf Y_k\right)}}-\sum_{l=1}^{\mathcal L}{\lambda_l\operatorname{P}_{l}^{\max}},
\end{align}
and the Lagrangian multiplier $\lambda_l\ge0$ is associated with the constraint in \eqref{equ16b} for the maximal transmitting power of the $l$-th AP and $\tilde\lambda=\left\{\lambda_l,\forall l\right\}$ is the collecting of all the dual variables.

From the first KKT condition in \eqref{equ18a}, the optimal solution of $\mathbf W_{l,k},\forall \left\{l,k\right\}$ is determined by
\begin{align}
    \mathbf W_{l,k}\left(\lambda_l\right)=\left(\mathbf A_l+\lambda_l\mathbf I_{\mathcal M_B}\right)^{\operatorname{-1}}\mathbf {\hat{H}}_{l,k}\mathbf Y_k\mathbf U_k, \forall l,k.
\end{align}

The dual variable $\lambda_l$ should be determined for ensuring the second KKT condition in \eqref{equ18b} are satisfied. Let $g_l\left(\lambda_l\right)=\sum_{k=1}^{\mathcal K}{\operatorname{Tr}\left(\mathbf W_{l,k}^{\operatorname{H}}\left(\lambda_l\right)\mathbf W_{l,k}\left(\lambda_l\right)\right)}-\operatorname{P}_{l}^{\max}$, if the inverse of $\mathbf A_l$ exists and $g_l\left(0\right)\le0$ holds, then the optimal active transmitting beamforming is given as 
\begin{align}
    \mathbf W_{l,k}^{\operatorname{\left(t+1\right)}}=\mathbf W_{l,k}\left(0\right), \forall l,k.
\end{align}
Otherwise, we have to find $\lambda_l$ for ensuring $g_l\left(\lambda_l\right)=0$, i.e.,
\begin{align}
  \sum_{k=1}^{\mathcal K}{\operatorname{Tr}\left(\mathbf W_{l,k}^{\operatorname{H}}\left(\lambda_l\right)\mathbf W_{l,k}\left(\lambda_l\right)\right)}=\operatorname{P}_{l}^{\max}.  
\end{align}
Note that $g_l\left(\lambda_l\right)$ is a monotonically decreasing function of $\lambda_l$, we can find $\lambda_l$ efficiently by employing the bisection search method \cite{6336836}.
The detailed information for optimizing the active transmitting beamforming matrices are summarized in Algorithm \ref{a1}. 
\begin{algorithm}[t]
\caption{Bisection Search for Solving Problem $\mathcal P_3$} 
\label{a1} 
\begin{algorithmic}[1]
\REQUIRE the bounds $\lambda_l^{\operatorname{ub}}$ and $\lambda_l^{\operatorname{lb}}$, $ \forall l$, threshold $\varepsilon$.
\STATE \textbf{If} $g_l\left(0\right)\le0$ holds, \textbf{output} $\mathbf W_{l,k}^{\operatorname{\left(t+1\right)}}\triangleq\mathbf W_{l,k}\left(0\right), \forall l,k$ and \textbf{terminate};
\STATE \textbf{Calculate} $\lambda_l=\left(\lambda_l^{\operatorname{ub}}+\lambda_l^{\operatorname{lb}}\right)/2$ and $g_l\left(\lambda_l\right)$;
\STATE \textbf{If} $g_l\left(\lambda_l\right)\le0$, set $\lambda_l^{\operatorname{ub}}=\lambda_l$, \textbf{otherwise}, set $\lambda_l^{\operatorname{lb}}=\lambda_l$;\\
\STATE \textbf{If} $\left|\lambda_l^{\operatorname{ub}}-\lambda_l^{\operatorname{lb}}\right|\le\varepsilon$, \textbf{output} $\lambda_l$ and $\mathbf W_{l,k}^{\operatorname{\left(t+1\right)}}\triangleq\mathbf W_{l,k}\left(\lambda_l\right), \forall l,k$, and then \textbf{terminate}.\\ \textbf{Otherwise}, go to step 2.
\end{algorithmic} 
\end{algorithm}

\subsection{Optimization of Passive Reflecting Beamforming}
In this subsection, we consider to optimize $\boldsymbol{\tilde\Theta}$, while fixing the auxiliary matrices $\mathbf {\tilde U}^{\left(\operatorname{t+1}\right)}$ and $\mathbf {\tilde Y}^{\left(\operatorname{t+1}\right)}$, and the active transmitting beamforming $\mathbf {\tilde W}^{\left(\operatorname{t+1}\right)}$. 
First, for the convenience of representation, we define the following matrix merging operations:
\begin{subequations}
\begin{align}
   \mathbf {W}_i&=\left[\mathbf W_{1,i}^{\operatorname{T}},\mathbf W_{2,i}^{\operatorname{T}},\cdots,\mathbf W_{\mathcal L,i}^{\operatorname{T}}\right]^{\operatorname{T}}\in \mathbb C^{\mathcal L \mathcal M_B\times \mathcal M_U},\\
   \mathbf {\hat D}_{k}&=\left[\mathbf {\hat D}_{1,k}^{\operatorname{T}},\mathbf {\hat D}_{2,k}^{\operatorname{T}},\cdots,\mathbf {\hat D}_{\mathcal L,k}^{\operatorname{T}}\right]^{\operatorname{T}}\in \mathbb C^{\mathcal L \mathcal M_B\times \mathcal M_U},\\
   \mathbf {\hat S}&=\left[\mathbf {\hat S}_{1},\mathbf {\hat S}_{2},\cdots,\mathbf {\hat S}_{\mathcal L}\right]\in \mathbb C^{\mathcal {\widetilde N}\times \mathcal L\mathcal M_B},
\end{align}
\end{subequations}
where $\mathcal {\widetilde N}\triangleq\mathcal R\mathcal N$. Then, by defining 
$\mathbf { W}=\sum_{i=1}^{\mathcal K}{\mathbf {W}_i \mathbf {W}_i^{\operatorname{H}}}$,
we have
\begin{subequations}\label{eq31}
\begin{align}
    \sum_{l=1}^{\mathcal L}{\mathbf {\hat H}_{l,k}^{\operatorname{H}}\mathbf W_{l,k}}&=\mathbf {\hat D}_{k}^{\operatorname{H}}\mathbf W_{k}+\mathbf {\hat G}_{k}^{\operatorname{H}}{\boldsymbol\Theta}\mathbf {\hat S} \mathbf W_{k},\\
\sum_{l=1}^{\mathcal L}\sum_{i=1}^{\mathcal K}{\mathbf {\hat H}_{l,k}^{\operatorname{H}}\mathbf W_{l,i}\mathbf W_{l,i}^{\operatorname{H}}\mathbf {\hat H}_{l,k}}
 &=\mathbf {\hat D}_{k}^{\operatorname{H}}\mathbf W \mathbf {\hat D}_{k}+\mathbf {\hat D}_{k}^{\operatorname{H}}\mathbf W \mathbf {\hat S}^{\operatorname{H}} {\boldsymbol\Theta}^{\operatorname{H}} \mathbf {\hat G}_k \notag\\
 &+\mathbf {\hat G}_k^{\operatorname{H}}{\boldsymbol\Theta}\mathbf {\hat S} \mathbf W\mathbf {\hat D}_k
 + \mathbf {\hat G}_k^{\operatorname{H}}{\boldsymbol\Theta}\mathbf {\hat S}\mathbf W \mathbf {\hat S}^{\operatorname{H}}{\boldsymbol\Theta}^{\operatorname{H}}\mathbf {\hat G}_k. 
\end{align}
\end{subequations}

By substituting \eqref{eq31} into $f_1\left(\mathbf {\tilde W}, \boldsymbol{\tilde\Theta}, \mathbf {\tilde U}, \mathbf {\tilde Y}\right)$ and omitting irrelevant constants with respect to $\boldsymbol{\tilde\Theta}$, i.e., $\sum_{k=1}^{\mathcal K}{\operatorname{Tr}\left(\mathbf {\bar U}_k \mathbf Y_k^{\operatorname{H}}\mathbf D_{k}^{\operatorname{H}}\mathbf W \mathbf D_{k}\mathbf Y_k\right)}$, $\sum_{k=1}^{\mathcal K}{\operatorname{Tr}\left(\mathbf {\bar U}_k \mathbf Y_k^{\operatorname{H}}\mathbf D_{k}^{\operatorname{H}}\mathbf W_k\right)}$, $\sum_{k=1}^{\mathcal K}{\operatorname{Tr}\left(\mathbf {\bar U}_k\mathbf W_k^{\operatorname{H}}\mathbf D_{k}\mathbf Y_k\right)}$, and $\sum_{k=1}^{\mathcal K}{\operatorname{Tr}\left(\sigma_k^2\mathbf {\bar U}_k\mathbf Y_k^{\operatorname{H}}\mathbf Y_k\right)}$, which have no impact on optimizing the passive reflecting beamforming at IRSs, we have the following simplified problem with respect to $\boldsymbol{\tilde\Theta}$:
\begin{subequations}\label{equ16}
\begin{align}
  \mathcal P_4:\;  \boldsymbol {\tilde \Theta}^{\left(\operatorname{t+1}\right)} \triangleq \arg\mathop{\min}_{\boldsymbol {\tilde \Theta}} \quad &\sum_{k=1}^{\mathcal K}{\operatorname{Tr}\left(\mathbf {\bar U}_k \mathbf Y_k^{\operatorname{H}} \mathbf {\hat G}_k^{\operatorname{H}}{\boldsymbol\Theta}\mathbf {\hat S}\mathbf W_k\right)}\notag\\
  &+\sum_{k=1}^{\mathcal K}{\operatorname{Tr}\left(\mathbf {\bar U}_k \mathbf W_k^{\operatorname{H}}\mathbf {\hat S}^{\operatorname{H}}{\boldsymbol\Theta}^{\operatorname{H}}\mathbf {\hat G}_k \mathbf Y_k\right)}\notag\\
    &-\sum_{k=1}^{\mathcal K}{\operatorname{Tr}\left(\mathbf {\bar U}_k \mathbf Y_k^{\operatorname{H}}\mathbf {\hat G}_k^{\operatorname{H}}{\boldsymbol\Theta}\mathbf {\hat S}\mathbf W \mathbf {\hat S}^{\operatorname{H}}{\boldsymbol\Theta}^{\operatorname{H}}\mathbf {\hat G}_k\mathbf Y_k\right)}\notag\\
    &-\sum_{k=1}^{\mathcal K}{\operatorname{Tr}\left(\mathbf {\bar U}_k \mathbf Y_k^{\operatorname{H}}\mathbf {\hat G}_k^{\operatorname{H}}{\boldsymbol\Theta}\mathbf {\hat S} \mathbf W\mathbf {\hat D}_k\mathbf Y_k\right)}\notag\\
    &-\sum_{k=1}^{\mathcal K}{\operatorname{Tr}\left(\mathbf {\bar U}_k \mathbf Y_k^{\operatorname{H}}\mathbf {\hat D}_{k}^{\operatorname{H}}\mathbf W \mathbf {\hat S}^{\operatorname{H}}{\boldsymbol\Theta}^{\operatorname{H}} \mathbf {\hat G}_k\mathbf Y_k\right)},\label{equ24a}\\
     \operatorname{s.t.} \quad& \left|\boldsymbol\Theta_{n,n}\right|=\alpha,\forall n\in \mathcal {\widetilde N}.\label{equ24b}
\end{align}
\end{subequations}

The above problem is still arduous to tackle. As follows, we transform problem $\mathcal P_4$ into an equivalent but more tractably form by employing some further algebraic  manipulations.

First, by defining $\mathbf Z_k = \mathbf {\hat G}_k \mathbf Y_k \mathbf {\bar U}_k \mathbf Y_k^{\operatorname{H}} \mathbf {\hat G}_k^{\operatorname{H}}$, $\mathbf Z =\sum_{k=1}^{\mathcal{K}}{\mathbf Z_k}$, and $\mathbf Q= \mathbf {\hat S} \mathbf W \mathbf {\hat S}^{\operatorname{H}}$, we have 
\begin{align}
   \sum_{k=1}^{\mathcal{K}}{\operatorname{Tr}\left(\mathbf {\bar U}_k \mathbf Y_k^{\operatorname{H}}\mathbf {\hat G}_k^{\operatorname{H}}\boldsymbol\Theta\mathbf {\hat S}\mathbf W \mathbf {\hat S}^{\operatorname{H}}\boldsymbol\Theta^{\operatorname{H}}\mathbf {\hat G}_k\mathbf Y_k\right)}=\operatorname{Tr}\left(\boldsymbol\Theta^{\operatorname{H}}\mathbf Z \boldsymbol\Theta \mathbf Q\right).
\end{align}
Then, by defining $\mathbf C_k =\mathbf {\hat G}_k\mathbf Y_k\mathbf {\bar U}_k\mathbf Y_k^{\operatorname{H}}\mathbf {\hat D}_k^{\operatorname{H}}\mathbf W\mathbf {\hat S}^{\operatorname{H}}$ and $\mathbf C =\sum_{k=1}^{\mathcal{K}}{\mathbf C_k}$, we have 
\begin{subequations}
\begin{align}
    \sum_{k=1}^{\mathcal{K}}{\operatorname{Tr}\left(\mathbf {\bar U}_k \mathbf Y_k^{\operatorname{H}}\mathbf {\hat D}_{k}^{\operatorname{H}}\mathbf W \mathbf {\hat S}^{\operatorname{H}}\boldsymbol\Theta^{\operatorname{H}} \mathbf {\hat G}_k\mathbf Y_k\right)}&=\operatorname{Tr}\left(\boldsymbol\Theta^{\operatorname{H}}\mathbf C\right),\\
    \sum_{k=1}^{\mathcal{K}}{\operatorname{Tr}\left(\mathbf {\bar U}_k \mathbf Y_k^{\operatorname{H}}\mathbf {\hat G}_k^{\operatorname{H}}\boldsymbol\Theta\mathbf {\hat S} \mathbf W\mathbf {\hat D}_k\mathbf Y_k\right)}&=\operatorname{Tr}\left(\mathbf C^{\operatorname{H}}\boldsymbol\Theta\right).
\end{align}
\end{subequations}
Next, by defining $\mathbf E_k=\mathbf {\hat G}_k\mathbf Y_k\mathbf {\bar U}_k\mathbf W_k^{\operatorname{H}}\mathbf {\hat S}^{\operatorname{H}}$ and $\mathbf E =\sum_{k=1}^{\mathcal{K}}{\mathbf E_k}$, we have 
\begin{subequations}
\begin{align}
    \sum_{k=1}^{\mathcal{K}}{\operatorname{Tr}\left(\mathbf {\bar U}_k \mathbf W_k^{\operatorname{H}}\mathbf {\hat S}^{\operatorname{H}}\boldsymbol\Theta^{\operatorname{H}}\mathbf {\hat G}_k \mathbf Y_k\right)}&=\operatorname{Tr}\left(\Theta^{\operatorname{H}}\mathbf E \right),\\
    \sum_{k=1}^{\mathcal{K}}{\operatorname{Tr}\left(\mathbf {\bar U}_k \mathbf Y_k^{\operatorname{H}} \mathbf {\hat G}_k^{\operatorname{H}}\boldsymbol\Theta\mathbf {\hat S}\mathbf W_k\right)}&=\operatorname{Tr}\left(\mathbf E ^{\operatorname{H}}\boldsymbol\Theta\right).
\end{align}
\end{subequations}
Equivalently, we have a sequence of equalities as follows
\begin{align}\label{equ30}
    \operatorname{Tr}\left(\boldsymbol\Theta^{\operatorname{H}}\mathbf Z \boldsymbol\Theta \mathbf Q\right)=\boldsymbol\theta^{\operatorname{H}}\boldsymbol{\mathcal Z}\boldsymbol\theta,\; \operatorname{Tr}\left(\boldsymbol\Theta^{\operatorname{H}}\boldsymbol\Omega \right)= \boldsymbol\theta^{\operatorname{H}}\boldsymbol\omega,
\end{align}
where $\boldsymbol{\mathcal Z} \triangleq \mathbf Z \odot \mathbf Q^{\operatorname{T}}\in \mathbb C^{\widetilde{\mathcal N}\times \widetilde{\mathcal N}}$, $\boldsymbol\Omega=\mathbf E - \mathbf C\in \mathbb C^{\widetilde{\mathcal N}\times \widetilde{\mathcal N}}$,
\begin{align}\label{equ31}
    \boldsymbol\theta=\operatorname{Vecd}\left(\boldsymbol\Theta\right)\in \mathbb C^{\widetilde{\mathcal N}\times 1},\; \boldsymbol\omega=\operatorname{Vecd}\left(\boldsymbol\Omega\right)\in \mathbb C^{\widetilde{\mathcal N}\times 1},
\end{align}
where $\operatorname{Vecd}\left(\mathbf X\right)$ forms a vector out of the diagonal of its matrix argument. The sequence of equalities given above in \eqref{equ30} follow from the properties in \cite[Theorem 1.11]{2017Matrix}.

Then, based on the aforementioned equivalent transformations, problem $\mathcal P_4$ can be reformulated by 
\begin{subequations}
\begin{align}
   \mathcal P_5:\; \boldsymbol\tilde\theta^{\operatorname{\left(t+1\right)}}\triangleq\mathop{\max}_{\boldsymbol\tilde\theta} &\quad -\boldsymbol\theta^{\operatorname{H}}\boldsymbol{\mathcal Z}\boldsymbol\theta+\boldsymbol\theta^{\operatorname{H}}\boldsymbol\omega+\boldsymbol\omega^{\operatorname{H}}\boldsymbol\theta\label{equ30a}\\
    \operatorname{s.t.} &\quad \left|\boldsymbol\theta_{n}\right|={\alpha},n =1,2,\cdots, \widetilde{\mathcal N}.\label{equ30b}
\end{align}
\end{subequations}

Due to the non-convexity of the constant modulus constraints in \eqref{equ30b}, the above optimization problem is still non-convex, and belongs to the class of NP-hard problems. The widely employed methods to solve problem $\mathcal P_5$ with high-quality suboptimal solution are employing semidefinte relaxation \cite{8930608} or quadratic relaxation \cite{8982186}, however, both the relaxation-based techniques suffer from high computational complexities, i.e., the order are larger than $\mathcal O\left(\mathcal R^6\mathcal N^6\right)$ \cite{8982186}. In practice, since $\mathcal R\times\mathcal {N}$ is large, it is computational impractical in real world. Hence, in the following, we provide a low complexity alternate sequential optimization (ASO) algorithm \cite{7946256} to solve this non-convex problem. 

Note that the objective function in \eqref{equ30a} and the constant modulus constraints in \eqref{equ30b} in problem $\mathcal P_5$ are separable with respect to $\boldsymbol\theta_i,\forall i \in \widetilde{\mathcal N}$ \cite{8930608,9110912}, therefore, we can decompose problem $\mathcal P_5$ into $\widetilde{\mathcal N}$ separate subproblems and solve them one-by-one. Particularly, we have 
\begin{align}
    \boldsymbol\theta^{\operatorname{H}}\boldsymbol\omega=\sum_{n=1}^{\widetilde{\mathcal N}}{\boldsymbol\theta_n^\ast\boldsymbol\omega_n}=\boldsymbol\theta_i^\ast\boldsymbol\omega_i +\sum_{n=1, n\ne i}^{\widetilde{\mathcal N}}{\boldsymbol\theta_n^{\ast}\boldsymbol\omega_n}.
\end{align}
Meanwhile, $\boldsymbol\theta^{\operatorname{H}} \boldsymbol{\mathcal { Z}} \boldsymbol\theta$ can be expanded as
\begin{align}\label{eq44}
     \boldsymbol\theta^{\operatorname{H}} \boldsymbol{\mathcal { Z}} \boldsymbol\theta&=\sum_{n=1, n\ne i}^{\widetilde{\mathcal N}}{ \boldsymbol\theta^{\operatorname{H}}\mathbf z _n \boldsymbol\theta_n}+{ \boldsymbol\theta^{\operatorname{H}}\mathbf z _i \boldsymbol\theta_i}\notag\\
     &=\sum_{n=1, n\ne i}^{\widetilde{\mathcal N}}{ \boldsymbol\theta_i^{\ast} z _{i,n} \boldsymbol\theta_n}+{ \boldsymbol\theta^{\operatorname{H}}\mathbf z _i \boldsymbol\theta_i}+\sum_{m=1, m\ne n}^{\widetilde{\mathcal N}}{\sum_{p=1, p\ne i}^{\widetilde{\mathcal N}}{ \boldsymbol\theta_m^{\ast} z _{m,p} \boldsymbol\theta_p}}\notag\\
    &= \boldsymbol\theta_i^{\ast} z _{i,i} \boldsymbol\theta_i+\sum_{n=1, n\ne i}^{\widetilde{\mathcal N}}{\left( \boldsymbol\theta_i^{\ast} z _{i,n} \boldsymbol\theta_n+ \boldsymbol\theta_i z _{n,i}\boldsymbol\theta_n^{\ast}\right)}+\sum_{m=1, m\ne n}^{\widetilde{\mathcal N}}{\sum_{p=1, p\ne i}^{\widetilde{\mathcal N}}{ \boldsymbol\theta_m^{\ast} z _{m,p} \boldsymbol\theta_p}},
\end{align}
where $ \boldsymbol{\mathcal Z} = \left[\mathbf z_1,\mathbf z_2,\cdots,\mathbf z_{\widetilde{\mathcal N}}\right]$ with $\mathbf z_n =\left[z_{1,n}, z_{2,n},\cdots,z_{\widetilde{\mathcal N},n}\right]^{\operatorname{T}}\in \mathbb C^{\widetilde{\mathcal{N}}\times 1}$. By using the property $z_{i,n}=z_{n,i}^{\ast}$ and basing on the fact that $ \boldsymbol{\mathcal { Z}}$ is a positive semi-definite matrix, we have
\begin{align}
     &-\boldsymbol\theta^{\operatorname{H}} \boldsymbol{\mathcal { Z}} \boldsymbol\theta
    +\boldsymbol\theta^{\operatorname{H}}\boldsymbol\omega+\boldsymbol\omega^{\operatorname{H}}\boldsymbol\theta\notag\\
    =& 2\operatorname{Re}\left\{\boldsymbol\theta_i^{\ast}\boldsymbol\omega_i+\sum_{n=1, n\ne i}^{\widetilde{\mathcal N}}{\boldsymbol\theta_n^{\ast}\boldsymbol\omega_n}-\sum_{n=1, n\ne i}^{\widetilde{\mathcal N}}{ \boldsymbol\theta_i^{\ast} z _{i,n} \boldsymbol\theta_n}\right\} 
    -\theta_i^{\ast} z _{i,i} \theta_i-\sum_{m=1, m\ne n}^{\widetilde{\mathcal N}}{\sum_{p=1, p\ne i}^{\widetilde{\mathcal N}}{ \boldsymbol\theta_m^{\ast} z _{m,p} \boldsymbol\theta_p}}\notag\\
    \triangleq&2\operatorname{Re}\left\{\boldsymbol\theta_i^{\ast}\boldsymbol\mu_i\right\}+\boldsymbol\xi
\end{align}
where 
\begin{subequations}
\begin{align}
    \boldsymbol\mu_{i}&=\boldsymbol\omega_i - \sum_{n=1, n\ne i}^{\widetilde{\mathcal N}}{ z _{i,n} \boldsymbol\theta_n},\\
    \boldsymbol\xi&=2\operatorname{Re}\left\{\sum_{n=1, n\ne i}^{\widetilde{\mathcal N}}{\boldsymbol\theta_n^{\ast}\boldsymbol\omega_n}\right\}-\sum_{m=1, m\ne n}^{\widetilde{\mathcal N}}{\sum_{p=1, p\ne i}^{\widetilde{\mathcal N}}{ \boldsymbol\theta_m^{\ast} z _{m,p} \theta_p}}-\boldsymbol\theta_i^{\ast} z _{i,i} \boldsymbol\theta_i,
\end{align}
\end{subequations}
where $\boldsymbol\xi$ is the irreverent constant term with regard to $\boldsymbol\theta_i$ (e.g., $\boldsymbol\theta_i^{\ast} z _{i,i} \boldsymbol\theta_i=z _{i,i}\left| \boldsymbol\theta_i\right|^2={\alpha^2}{z _{i,i}}$), which do not affect the optimal solution.
Thus, we can only investigate $\operatorname{Re}\left\{ \boldsymbol\theta_i^{\ast}\boldsymbol\mu_i\right\}$ for sequentially optimizing each PS while fixing the remaining $\widetilde{\mathcal N}-1$ PSs. Hence, we have
\begin{subequations}
\begin{align}
   \mathcal P_6:\; \mathop{\max}_{ \boldsymbol\theta_i}\quad& \operatorname{Re}\left\{ \boldsymbol\theta_i^{\ast}\boldsymbol\mu_i\right\}\\
    \operatorname{s.t.}\quad&\left| \boldsymbol\theta_i\right|={\alpha},
\end{align}
\end{subequations}

An equivalent expression for the above problem is given by
\begin{subequations}
\begin{align}
  \mathcal P_7:\;  \mathop{\max}_{ \phi_i}\quad& \cos\left(- \phi_i+\eta_i\right)\\
    \operatorname{s.t.}\quad& \phi_i\in\left[0,2\pi\right],
\end{align}
\end{subequations}
where $\eta_i$ and $- \phi_i$ are the phase-shifting of $\boldsymbol\mu_i$ and $ \boldsymbol\theta_i^{\ast}$, respectively. Thus, problem $\mathcal P_7$ has a closed-form optimal, which is given by
\begin{align}
     \phi_i=\eta_i, \forall i \in \widetilde{\mathcal N}.
\end{align}
Accordingly, we have 
\begin{align}\label{equ38}
     \boldsymbol\theta_i={\alpha}e^{j\eta_i}, \forall i \in \widetilde{\mathcal N}.
\end{align}

Based on the above discussions, the procedure of sequentially optimizing $ \boldsymbol\theta_1, \boldsymbol\theta_2,\cdots, \boldsymbol\theta_{\widetilde{\mathcal N}}$ and then repeatedly until the convergence is attained. The details for optimizing the locally optimal passive reflecting beamforming by employing the ASO algorithm are summarized in Algorithm \ref{a2}, and we have the following lemma
\begin{lemma}\label{lemmaconver}
Algorithm \ref{a2} is guaranteed to converge.
\end{lemma}
\begin{proof}
The proof is presented in Appendix \ref{lemma2proof}
\end{proof}

\begin{algorithm}[t]
\caption{ASO Algorithm for Optimizing the Passive Reflecting Beamforming} 
\label{a2} 
\begin{algorithmic}[1] 
\REQUIRE $ \boldsymbol{\mathcal{Z}}$, threshold $\varepsilon$, $ \boldsymbol\theta^{\operatorname{0}}\triangleq\boldsymbol\theta^{\left(\operatorname{t}\right)}$, $u=1$.
\STATE \textbf{Sequentially Calculate} $ \boldsymbol\theta_i^{\operatorname{u}},\forall i \in {\widetilde{\mathcal N}}$, by using \eqref{equ38};
\STATE \textbf{Calculate} $\rho^{\operatorname{u}}=-\left({\boldsymbol\theta^{\operatorname{u}}}\right)^{\operatorname{H}}\boldsymbol{\mathcal Z}\boldsymbol\theta+\left({\boldsymbol\theta^{\operatorname{u}}}\right)^{\operatorname{H}}\boldsymbol\omega+\boldsymbol\omega^{\operatorname{H}}\boldsymbol\theta^{\operatorname{u}}$;
\STATE \textbf{If} $\left|\rho^{\operatorname{u+1}}-\rho^{\operatorname{u}}\right|\le\varepsilon$, \textbf{output} $ \boldsymbol\theta^{\left(\operatorname{t+1}\right)}\triangleq \boldsymbol\theta^{\operatorname{u}}$ and then \textbf{terminate};\\
\textbf{Otherwise}, update $\operatorname{u} \leftarrow \operatorname{u+1}$ and go to Step 1.
\end{algorithmic} 
\end{algorithm}

\subsection{Analysis of the Overall Algorithm}
\begin{algorithm}[t]
\caption{BCD-based RJD Algorithm for Solving Problem $\mathcal P_1$} 
\label{a3} 
\begin{algorithmic}[1]
\REQUIRE    
    $\mathbf {\tilde W}^{\left ( \operatorname{0} \right )}$,
    $\boldsymbol{\tilde\Theta} ^{\left ( \operatorname{0} \right )}$;
    threshold $\varepsilon$, $\operatorname{t}=0$.
\STATE \textbf{Calculate} $\mathbf {\tilde Y} ^{\left ( \operatorname{t+1} \right )}$ by using \eqref{equ14};
\STATE \textbf{Calculate} $\mathbf {\tilde U} ^{\left ( \operatorname{t+1} \right )}$, by using \eqref{equ15};
\STATE \textbf{Calculate} $\mathbf {\tilde W} ^{\left ( \operatorname{t+1} \right )}$, by employing Algorithm \ref{a1};
\STATE \textbf{Calculate} $\boldsymbol{\tilde \Theta}^{\left ( \operatorname{t+1}\right )}$, by employing Algorithm \ref{a2};\\
\textbf{If} {${{\left| {{\mathcal{R}^{\left( \operatorname{t+1} \right)}}- {\mathcal{R}^{\left( \operatorname{t}\right)}}} \right|}} < \varepsilon$}, \textbf{output} $\mathbf {\tilde W}_{l,k}^{\operatorname{opt}} \triangleq \mathbf {\tilde W}  ^{\left ( \operatorname{t} \right )}$ and $\boldsymbol{\tilde \Theta}^{\operatorname{opt}} \triangleq \boldsymbol{\tilde \Theta}^{\left (  \operatorname{t} \right )}$, and then \textbf{terminate};
\textbf{Otherwise}, update $\operatorname{t}\leftarrow\operatorname{t+1}$, and go to step 1.
\end{algorithmic} 
\end{algorithm}
The details of the proposed BCD-based robust joint design (RJD) algorithm are summarized in Algorithm \ref{a3}. 
The convergence of the proposed RJD algorithm can be guaranteed since the average sum-rate $\mathcal{R}\left(\mathbf {\tilde W},\boldsymbol{\tilde\Theta}  \right)$ is monotonically non-decreasing after each updating step in Algorithm \ref{a3}. Overall, it can be summarized that 
\begin{align}
    \mathcal{R}\left( {{\mathbf {\tilde W}^{\left(\operatorname{t+1}\right)}},{\boldsymbol{\tilde\Theta}^{\left(\operatorname{t+1}\right)} }} \right)\ge \mathcal{R}\left( {{\mathbf {\tilde W}^{\left(\operatorname{t+1}\right)}},{\boldsymbol{\tilde\Theta}^{\left(\operatorname{t}\right)} }} \right)\ge\mathcal{R}\left( {{\mathbf {\tilde W}^{\left(\operatorname{t}\right)}},{\boldsymbol{\tilde\Theta}^{\left(\operatorname{t}\right)} }} \right)
    \ge \cdots \ge \mathcal{R}\left( {{\mathbf {\tilde W}^{\left(\operatorname{1}\right)}},{\boldsymbol{\tilde\Theta}^{\left(\operatorname{1}\right)} }} \right).
\end{align}
Besides, $\mathcal{R}\left( {{\mathbf {\tilde W}},{\boldsymbol{\tilde\Theta}}} \right)$ is upper-bounded by a finite value due to the limited transmit powers of APs and a finite number of PSs of IRSs. As the number of iterations increases, we finally have $ \mathcal{R}\left( {{\mathbf {\tilde W}^{\operatorname{opt}}},{\boldsymbol{\tilde\Theta}^{\operatorname{opt}} }} \right)\triangleq\mathcal{R}\left( {{\mathbf {\tilde W}^{\left( \operatorname{t}_{\max}\right)}},{\boldsymbol{\tilde\Theta}^{\left(\operatorname{t}_{\max}\right)} }} \right)$, where $\operatorname{t}_{\max}$ is the maximal number of iterations when Algorithm \ref{a3} converges.

Meanwhile, the complexity of Algorithm \ref{a3} is briefly discuss as follows. The complexities of updating $\mathbf U_k$ and $\mathbf Y_k$ are on the order of $\mathcal O \left(\mathcal M_U^3\right), \forall k$, respectively. The complexity of calculating $\mathbf W_{l,k}$ is $\mathcal O\left(\mathcal M_B^3\right), \forall {l,k}$, and the complexity of updating $\boldsymbol\Theta$ is $\mathcal O\left(\mathcal {R}^2\mathcal {N}^2\right)$. Therefore, based on the aforementioned discussions, the total complexity of Algorithm \ref{a3} is $\mathcal O\left(\mathcal I_{3}\left(2\mathcal K\mathcal M_U^3+\mathcal I_{1}\mathcal L\mathcal K\mathcal M_B^3+\mathcal I_{\operatorname{2}}\mathcal {R}^2\mathcal {N}^2\right)\right)$, where $\mathcal I_{1}$, $\mathcal I_{\operatorname{2}}$ and $\mathcal I_{3}$ are the number of iterations when Algorithms \ref{a1}, \ref{a2},  and \ref{a3} converge, respectively. 

\section{Numerical Simulation and Discussion}{\label{sec4}}
As follows, simulation results are provided to evaluate the performance of the proposed BCD-based RJD algorithm. 
We consider a system schematic as shown in Fig. \ref{f2}, a three-IRSs assisted cell-free wireless communication network, where six-APs are transmitting signals to four-UEs cooperatively. 
We assume a 3-D scenario, where the height of the APs, the IRSs, and the UEs are 3m, 6m, and 1.5m, respectively. The four UEs are uniformly and randomly distributed in a circle centered at $(\chi,100)$ with a radius of 10m and the locations of APs and IRSs are shown in Fig. \ref{f2}. The distance-dependent large-scale path loss is denoted by
\begin{align}
    L\left(d_{x}\right)= \mathcal C_0 \left({d_x}/{d_0}\right)^{-p_x},
\end{align}
where $\mathcal C_0=-30$ dB is the channel power gain at the reference distance $d_0=1$m, and $d_x$ and $p_x$ denote the link distance and path loss exponent. We set $p_{D_{l,k}}=3.75,\forall {l,k}$, $p_{S_{l,r}}=p_{G_{r,k}}=2.2,\forall {l,k,r}$. 
Meanwhile, we assume the PSs along the horizontal and vertical are $\mathcal N_h=10$ and $\mathcal N_v=\mathcal N/\mathcal N_h$, respectively. The noise variance at each UE is $\sigma_k^2=\sigma^2=-90$ dBm. For the statistical CSI error model, the normalized CSI errors are set as $\kappa^2=\kappa_{\mathbf D}^2=\kappa_{\mathbf S}^2=\kappa_{\mathbf F}^2=0.001$.
Unless otherwise stated, the other simulation parameters used in the rest of this section are setting as $\mathcal M_B=4$, $\mathcal M_U=2$, $\mathcal N=20$, $\alpha=1$, $P_{l}^{\max}=0.1$ W, and the Rician factors are $\beta_{G}=\beta_{S}=3$. 

Furthermore, in practice, the phase-shifting of each PS can only take finite discrete value from the discrete phase-shifting set \cite{8930608,9279253}. Let $b$ denote the number of bits to represent the resolution levels of IRS, the discrete phase-shifting set is represented by
\begin{align}
   \phi_{{\mathcal D}}\in \mathcal D\triangleq \left\{\frac{2\pi g}{2^b} \left| g=0, 1, 2, \cdots, 2^b-1\right.\right\}. 
\end{align}
To our best knowledge, the only way to obtain the global optimal from the discrete set is to perform exhaustive search, and the computational complex is $\mathcal O \left(2^{b\mathcal R\mathcal {N}}\right)$, i.e., exponential increasing with respect to the number of PSs of IRSs $\mathcal R\times\mathcal {N}$ and the IRS resolution $b$ bits. In practice, since $\mathcal R\times\mathcal {N}$ is large, it is computational impractical in real world. As a compromising approach, the common way to obtain the local optimal discrete phase-shifting of PS is projecting the continuous phase-shifting $\boldsymbol \phi_{n}^{\operatorname{\mathcal C}}$ obatianed by \eqref{equ38} to the nearest discrete value in the set $\mathcal D$, which has been widely employed in existing works \cite{9389801}.
It is computational efficient, i.e., $\mathcal O \left(\mathcal R^2\mathcal N^2 \times 2^b\right)$, and is much smaller than $\mathcal O \left(2^{b\mathcal R\mathcal {N}}\right)$. 

In the following simulations, we investigate the proposed RJD algorithm under the continuous phase-shifting case and the discrete phase-shifting case of 2-bits and 1-bit, i.e., \textit{{RJD, Continuous}}, \textit{{RJD, 2-bits}}, and \textit{{RJD, 1-bit}}. Meanwhile, all the simulation results are obtained by averaging 500 channel realizations. In the rest section, we investigate the convergence behavior and the impacts of the critical simulation parameters on the performance gains achieved by the proposed RJD algorithms as mentioned above, including {\textit{RJD, Continuous}}, {\textit{RJD, 2-bits}}, and {\textit{RJD, 1-bit}}.

\begin{figure}
    \centering
 \includegraphics[width=0.5\linewidth]{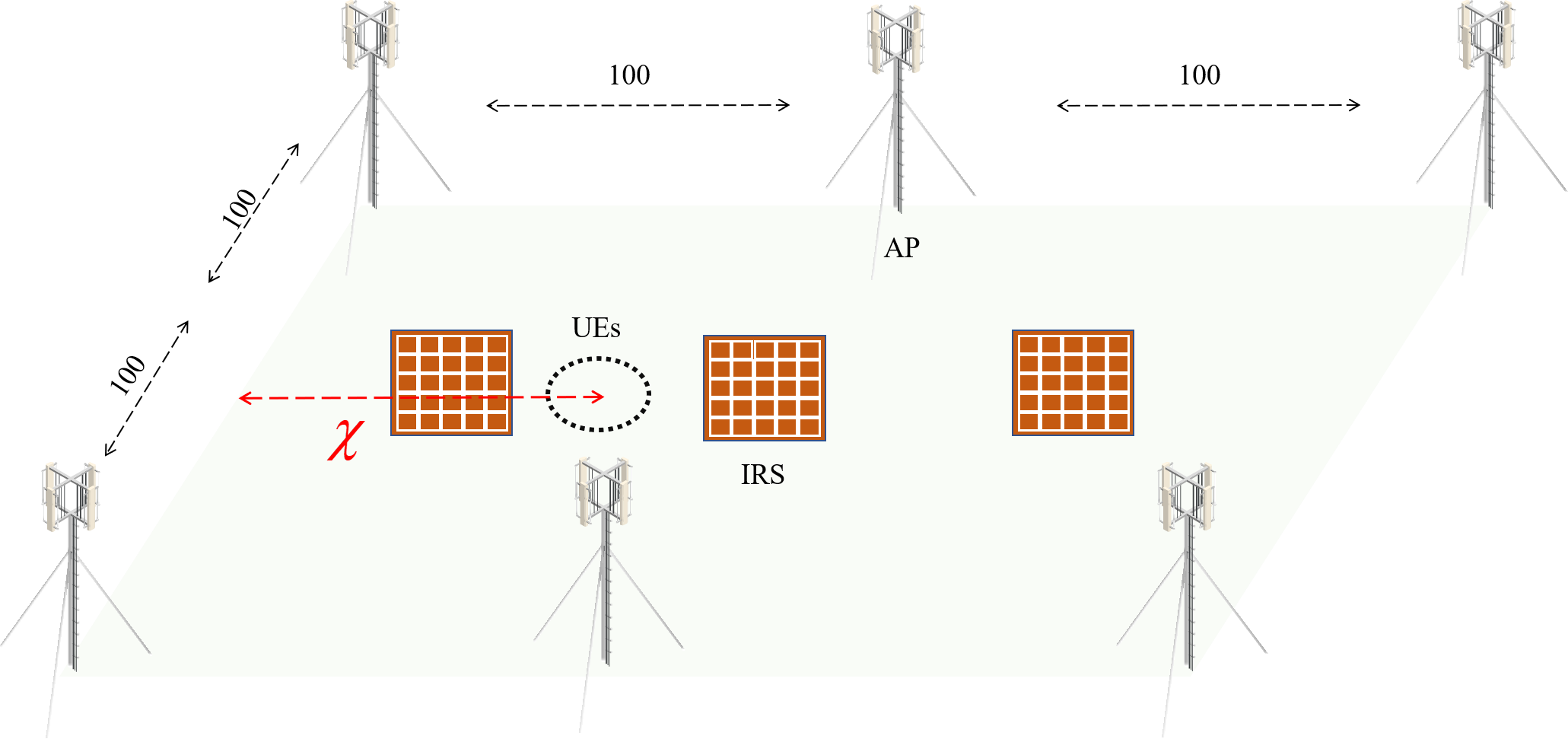}
    \caption{A location schematic of the considered IRSs-assisted cell-free network scenario.}
    \label{f2}
\end{figure}

\subsection{Convergence Behavior}
In this subsection, we investigate the convergence behavior of the proposed RJD algorithms.
As shown in Fig. \ref{fconvergence}, we plot the curves of the average sum-rate achieved by the proposed RJD schemes against the number of iterations under the different number of PSs at each IRS, i.e., $\mathcal N=20$ and $\mathcal N=50$. 
The curves are consistent with our expectation, where the three schemes converge to stationary points after a few iterations. 
It is observed that the convergence speed is sensitive to the number of PSs at each IRS, where the proposed RJD schemes under $\mathcal N=50$ is slower than that under $\mathcal N=20$.
Besides, since the size of the discrete phase-shifting set $2^b$ is much small than $\mathcal R^2\mathcal N^2$, thus it has a limited impact on the convergence speed, and the curves of \textit{{RJD, 2-bits}} and \textit{{RJD, 1-bit}} schemes also prove this analysis.
\begin{figure}
    \centering
    {\includegraphics[width=0.48\linewidth]{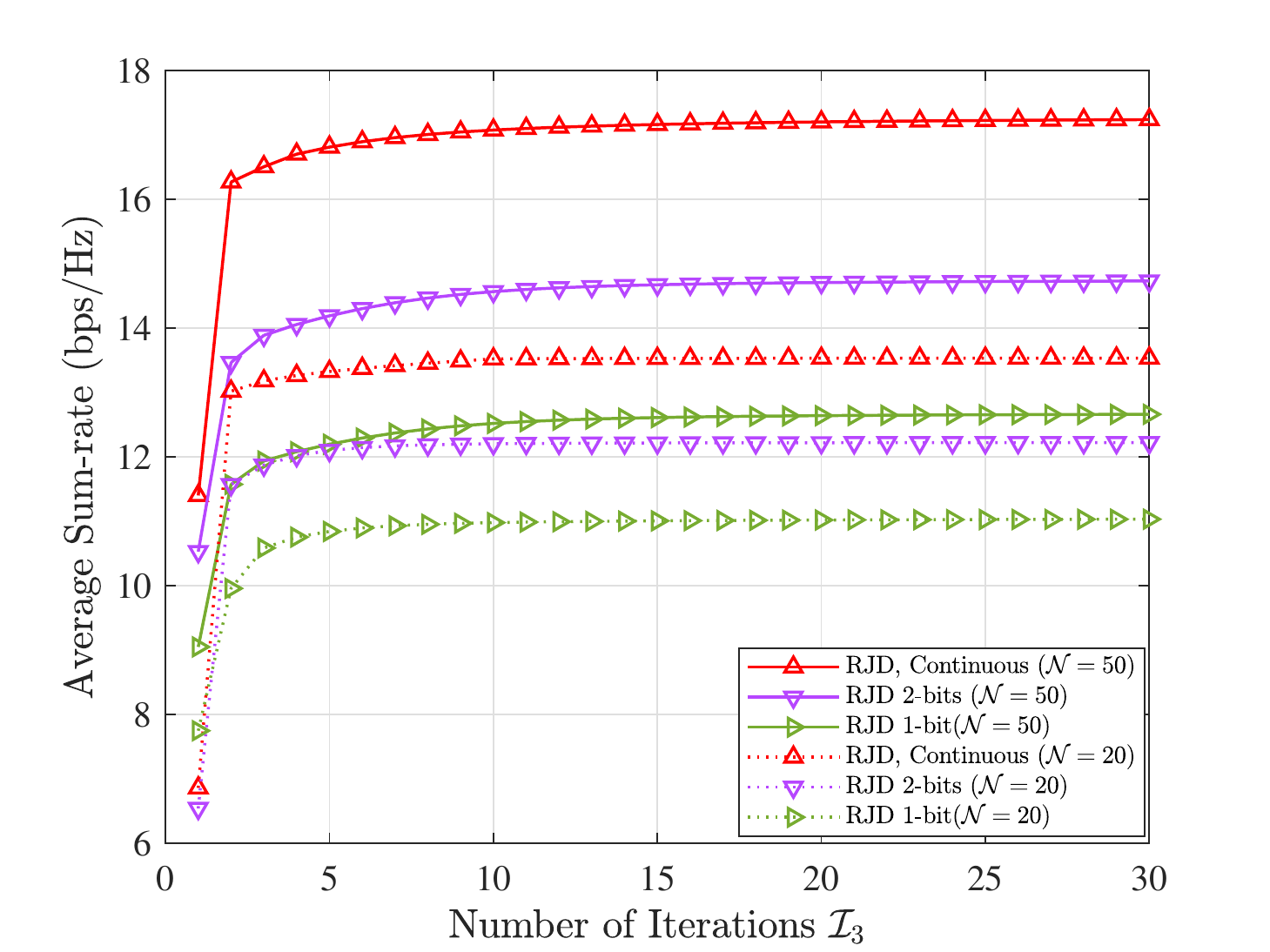}}
    \caption{Average sum-rate against the number of iterations.}
 \label{fconvergence} 
\end{figure}

\subsection{Performance Comparison and Discussions}
To evaluate the performance gain achieved by the IRS, we introduce the conventional cell-free scheme without the aid of IRSs \cite{6515204}, which can be solved by employing Algorithm \ref{a3} after setting the reflecting efficiency of all IRSs as zeros, i.e., $\alpha=0$, and is denoted by {\textit{Conventional CF}}. Additional, {\textit{Upper Bound}} scheme is also considered to denote the perfect CSI case.

\subsubsection{Impact of the CSI estimation error }\label{imcsi}First, we study the impact of the normalized CSI error covariances on the average sum-rate. 
As shown in Fig. \ref{fCSIerrora}, we set the normalized CSI errors of all channels are equal, i.e., $\kappa_{\mathbf D}^2=\kappa_{\mathbf S}^2=\kappa_{\mathbf F}^2=\kappa^2$, and investigate the average sum-rate achieved by all schemes against the normalized CSI error. 
It can be observed that with increasing $\kappa^2$, the performance gaps between all schemes with {\textit{Upper Bound}} become larger, which is consistent with the expectation. Besides, the performance loss of the proposed RJD schemes (include {\textit{RJD, Continuous}}, {\textit{RJD, 2-bits}}, and {\textit{RJD, 1-bit}}) are larger than {\textit{Conventional CF}} with respect to $\kappa^2$, which is due to the channel dimensions in the conventional cell-free network is much smaller than that in the IRS-assisted cell-free network under the considered system setting.

Then, we study the impact of each type of the channel with the different normalized CSI error level on the average sum-rate. As shown in Fig. \ref{fCSIerrorb}, we plot the average sum-rate achieved by the proposed {\textit{RJD, Continuous}} scheme with the different number of PSs.
It can be observed that the performance loss is sensitive to the number of PSs of each IRS. Most importantly, the normalized CSI error of the channel from the AP to the IRS, i.e., $\mathbf S_{l,r}, \forall l, r$ has the largest impact on the average sum-rate due to that the number of antennas at each AP and the number of PSs at each IRS are 
large, and the aggregate CSI uncertainty increase alongside the corresponding channel dimension. Therefore, as the normalized CSI error becomes
larger, the performance gap gradually expands. Fortunately, due to the positions of the APs and the IRSs are fixed, and so the CSI knowledge between the APs and the IRSs can be estimated with satisfactory accuracy by employing the angles of arrival and departure \cite{9366805}. 

\begin{figure}
    \centering
    \subfigure[Equal normalized CSI error.]{
    {\includegraphics[width=0.45\linewidth\label{fCSIerrora}]{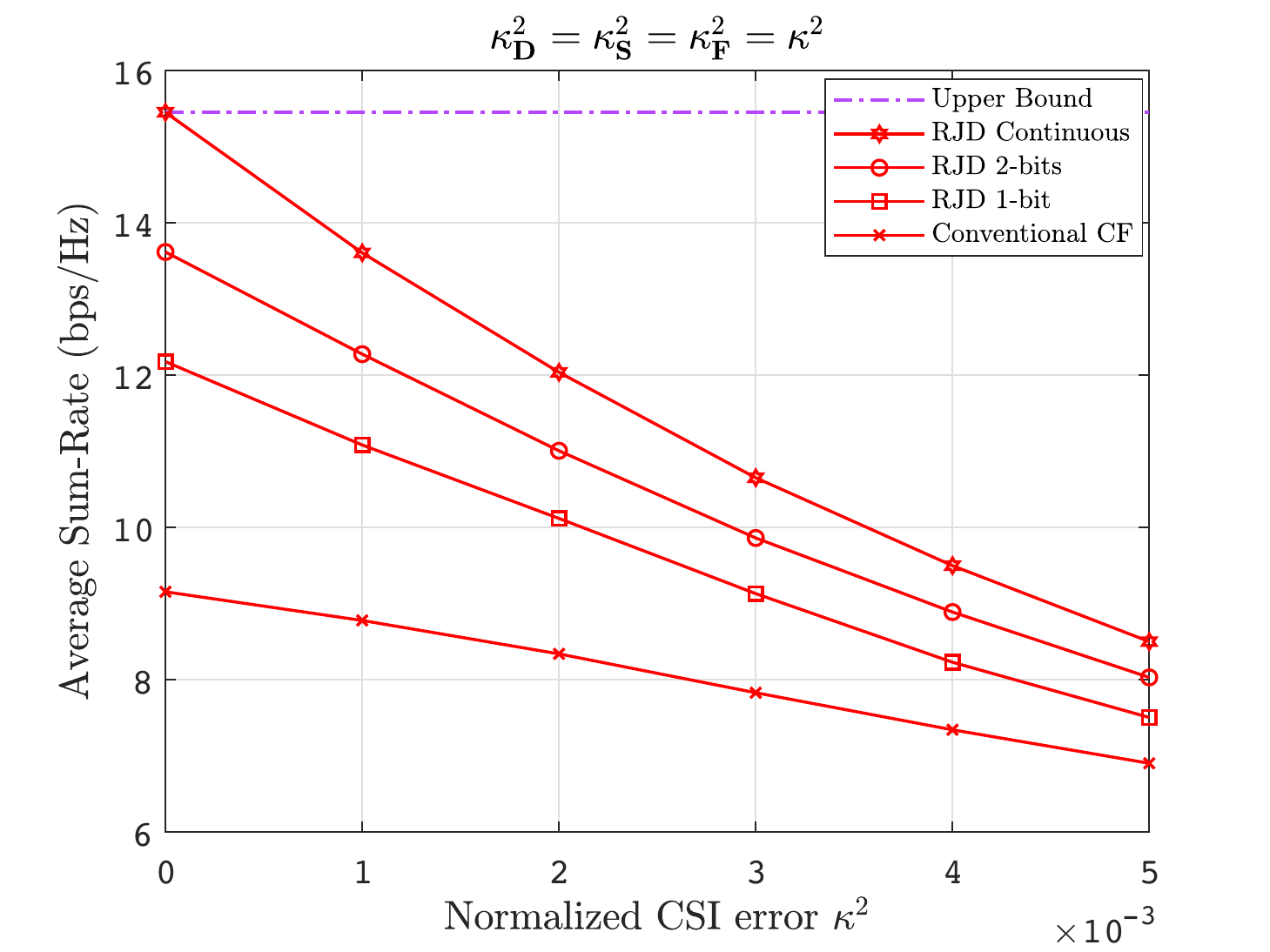}}}
    \subfigure[Different normalized CSI error.]{{\includegraphics[width=0.45\linewidth\label{fCSIerrorb}]{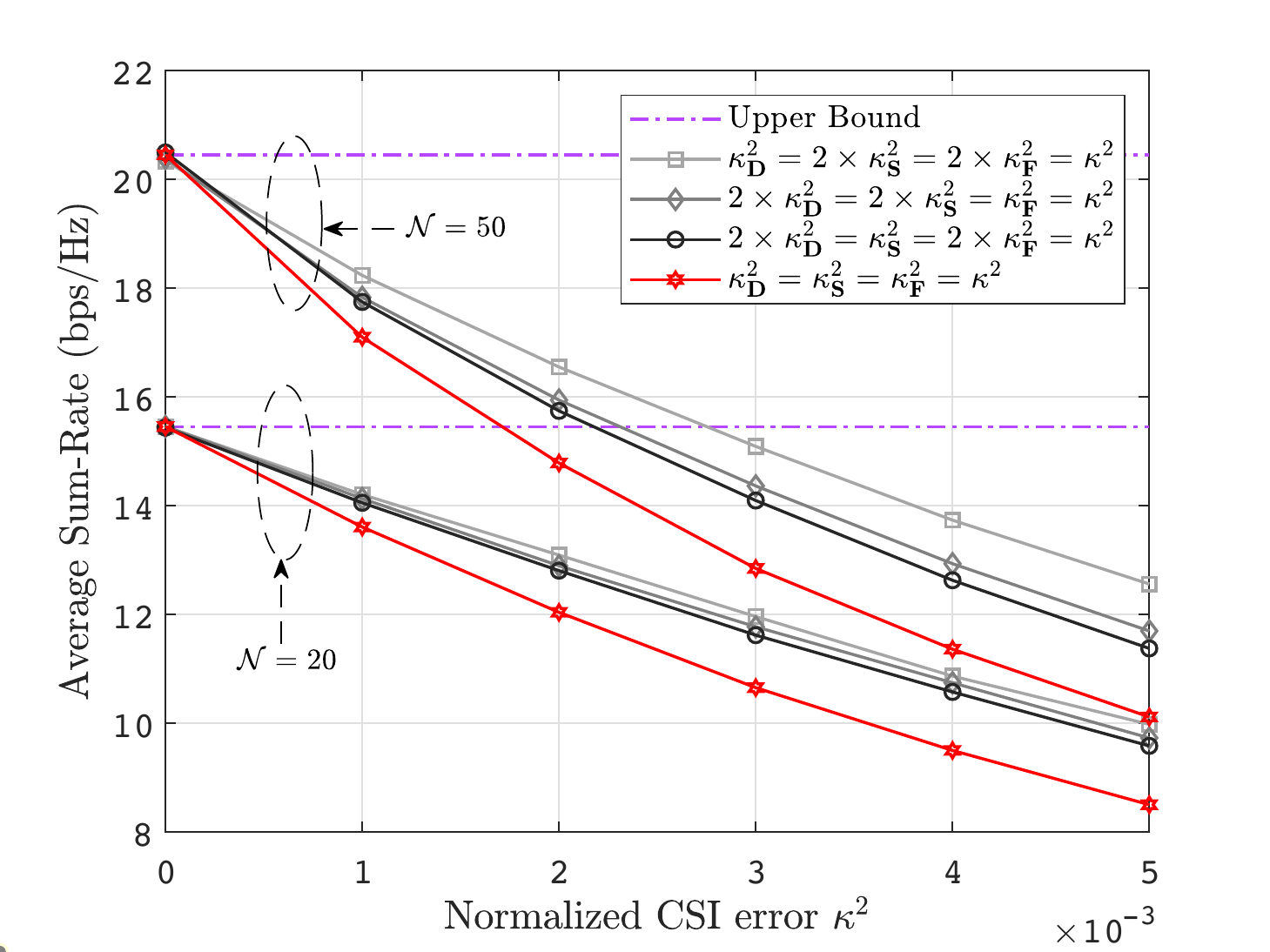}}}
    \caption{Average sum-rate against the normalized CSI error.}
\end{figure}

\subsubsection{Impact of the number of PSs at each IRS}
We present the average sum-rates achieved by all schemes against the number of the PSs at each IRS under a pair of normalized CSI errors in Fig. \ref{fnirs}.
It is illustrates that IRSs can considerably improve the performance compared with the conventional cell-free system, i.e., {\textit{Conventional CF}}, wherewith the increasing number of PSs, the performance achieved by the IRS-related schemes (including {\textit{Upper Bound}}, {\textit{RJD, Continuous}}, {\textit{RJD, 2-bits}}, and {\textit{RJD, 1-bit}}) increase. The trend of the curves highlights the importance of utilizing IRSs.
Besides, the size of the discrete phase-shifting set impacts the performance, and as expected, with increasing $\mathcal N$, the performance gaps between {\textit{RJD, 2-bits}} and {\textit{RJD, 1-bit}} become larger. 
The performance loss compared with {\textit{RJD, Continuous}} can be compensated by adopting high-resolution discrete PSs.
Most importantly, while the overall trends of the curves resemble those of expectation, an interesting observation can be made. Particularly,
when $\mathcal N$ is increased, although the performance gaps between the proposed RJD schemes and {\textit{Upper Bound}} become large, the slopes of the proposed RJD schemes curves towards decrease, which can be attributed to that with the imperfect CSI knowledge, a large number of PSs also means a large CSI uncertainty, thus the relative performance gain also decreased.

\begin{figure}
    \centering
    {\includegraphics[width=0.45\linewidth]{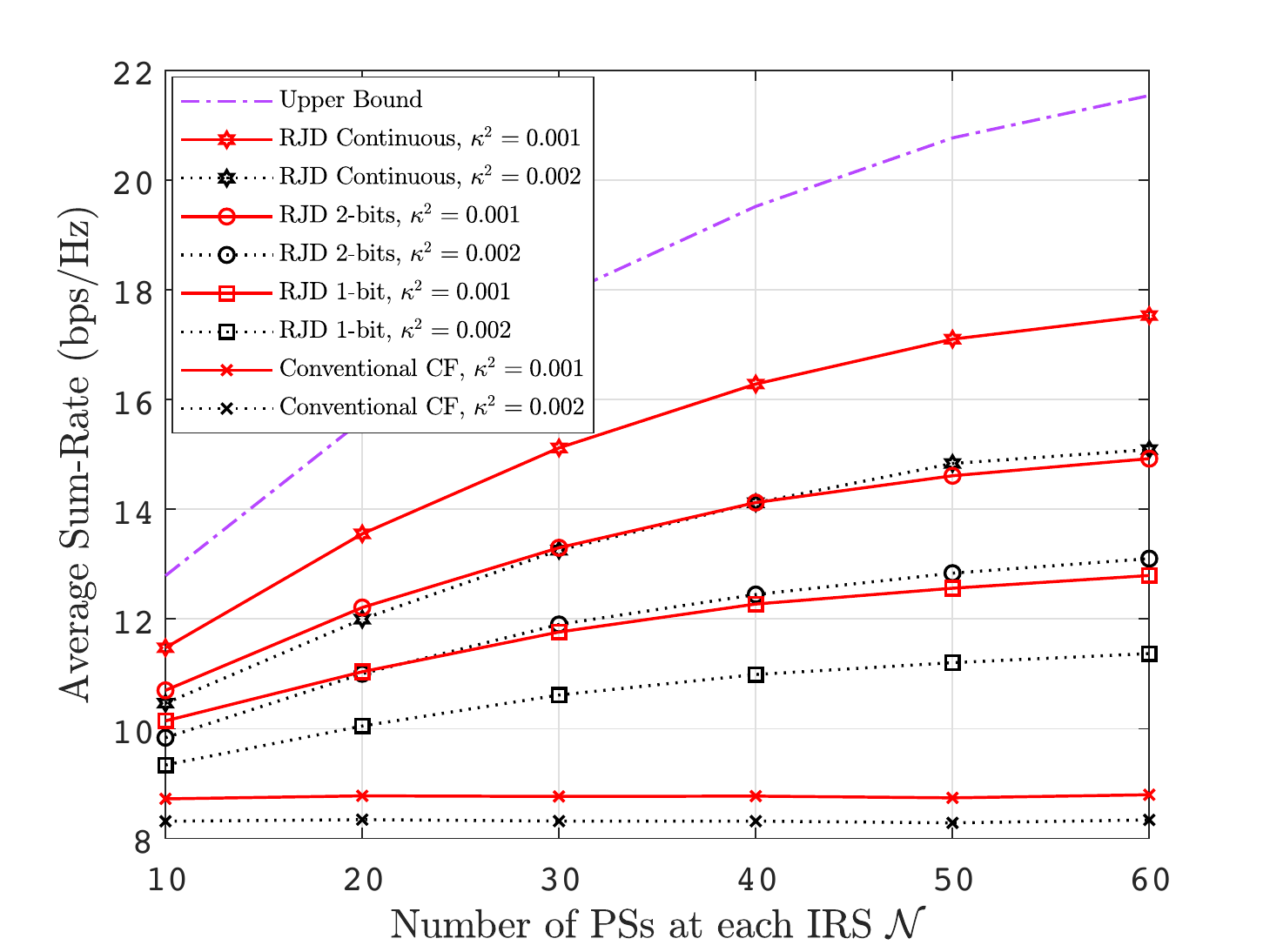}}
    \caption{Average sum-rate against the number of PSs at each IRS.}
 \label{fnirs} 
\end{figure}

\subsubsection{Impact of the UEs location}
We investigate the impact of the UEs location in Fig. \ref{firsloc} while the locations of APs and UEs remain fixed. The curves are approximately symmetric with respect to the line of $\chi=100$m, which are consistent with the expectation. 
It is observed that as the UEs are deployed closer to each IRS, i.e., $\chi=50$m, $\chi=100$m, and $\chi=150$m, the IRS-related schemes achieve the best performances, which is due to the smaller reflection channel fading. This indicates that the system performance can indeed be improved significantly with the deployment of IRSs, especially when the UEs are closed to the IRS.

\begin{figure}
    \centering
    {\includegraphics[width=0.45\linewidth]{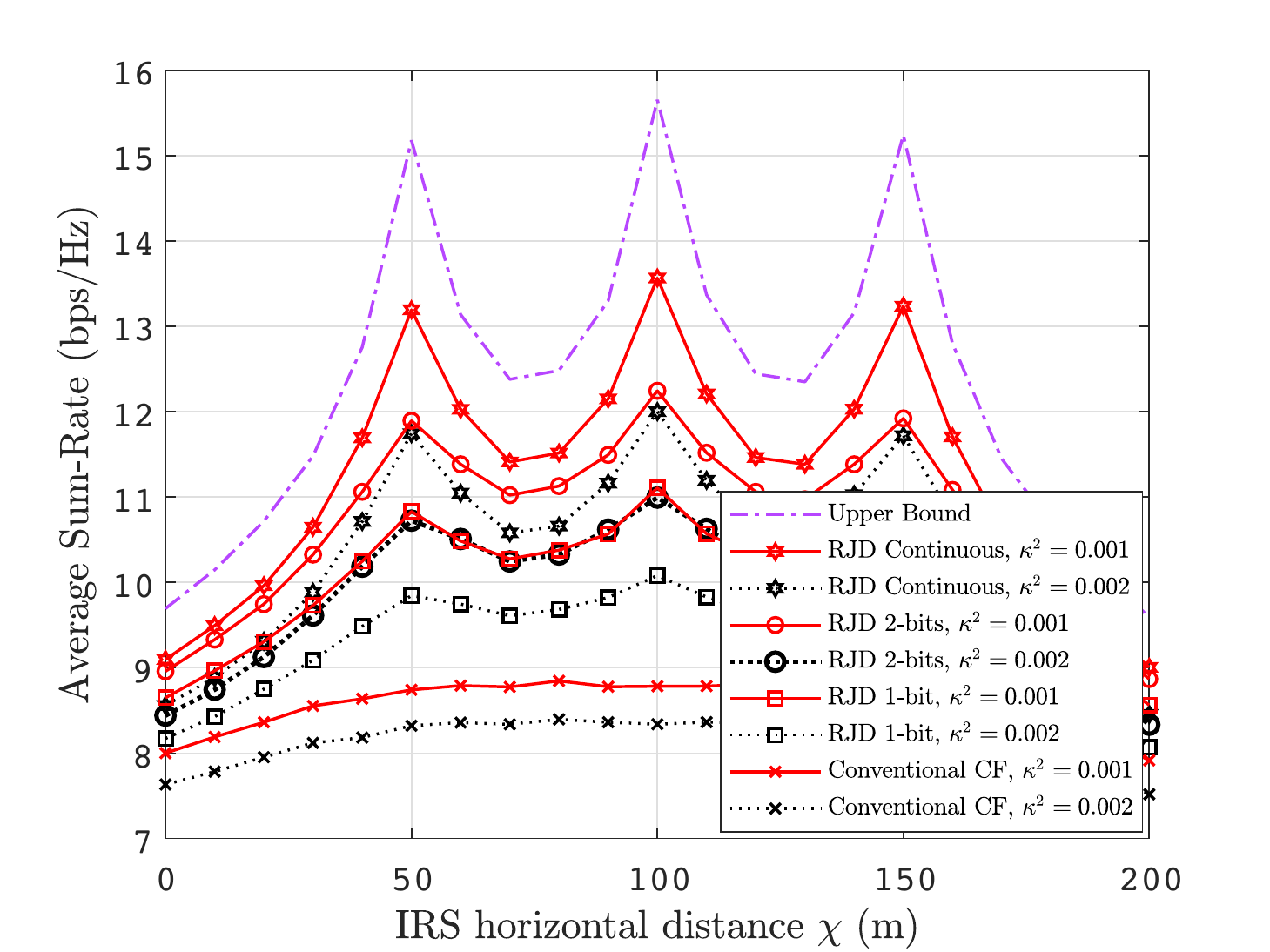}}
    \caption{Average sum-rate against the location of UEs.}
    \label{firsloc}
\end{figure}

\subsubsection{Impact of the path loss exponent}
We investigate the impact of the path loss exponent of the IRS-related channels while fixing that of direct channels. 
As shown in Fig. \ref{fpathloss}, the average sum-rate achieved by all IRS-related schemes decrease significantly with the increasing of the path loss exponent, and finally (i.e., $\rho \ge 3.6$), the curves are approximately coinciding with {\textit{Conventional CF}}. 
This is mainly due to that when the path loss exponents of IRS-related channels are large, the array gains introduced by IRS are negligible. 
To this end, the location of the IRS should be appropriately chosen for ensuring a free space IRS-related channels can be established, which thus to establish reliable communication links between APs and UEs.
\begin{figure}
    \centering
    {\includegraphics[width=0.45\linewidth]{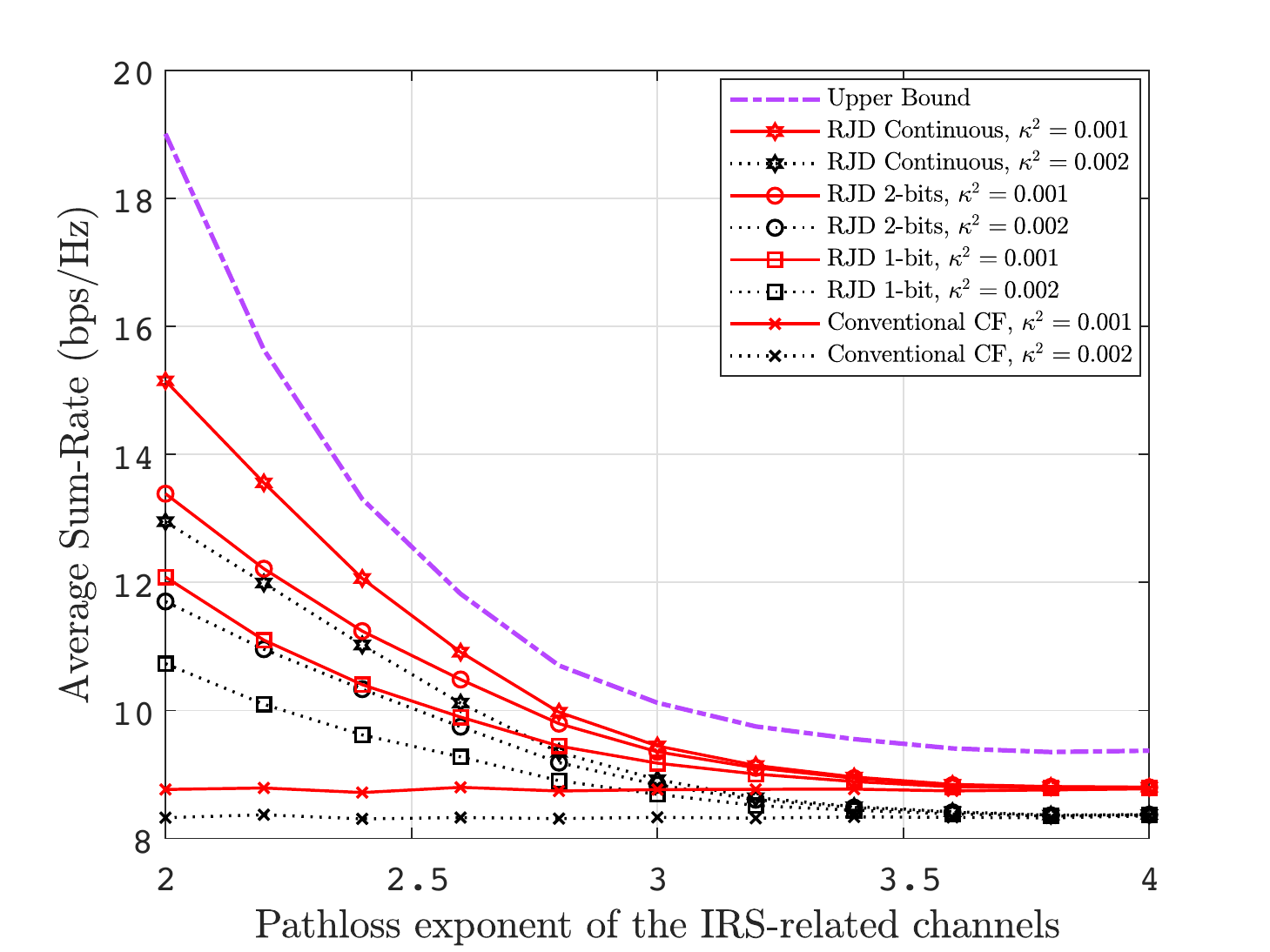}}
    \caption{Average sum-rate against the pathloss exponent of the IRS-related channels.}
 \label{fpathloss} 
\end{figure}

\subsubsection{Impact of the reflecting efficiency}
Fig. \ref{efficiency} illustrates the average sum-rate against the reflecting efficiency of IRSs. The comparative performance patterns are similar to Fig. \ref{fnirs}.
It can be observed that the reflecting efficiency of IRS has a substantial impact on the performance, where as expected, with the increasing $\alpha$, the average sum-rate achieved by IRS-related schemes increased significantly. 
It can be attributed to that a larger $\alpha$ means the fewer power loss caused by signal absorption at IRSs.
\begin{figure}
    \centering
    {\includegraphics[width=0.45\linewidth]{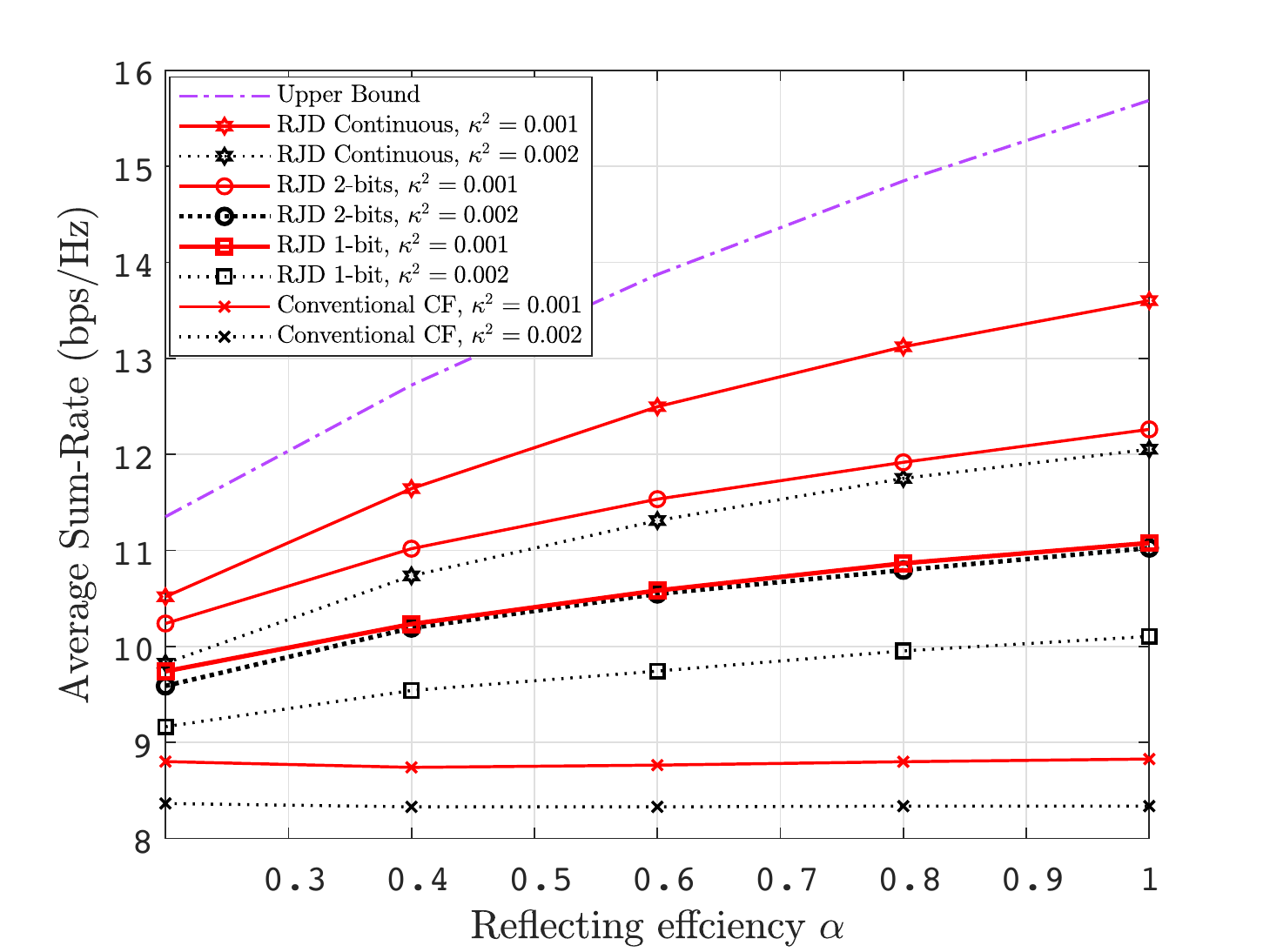}}
    \caption{Average sum-rate against the reflecting efficiency of IRSs.}
 \label{efficiency} 
\end{figure}

\section{Conclusion}{\label{sec5}}
In this paper, we investigated the robust design for IRS-assisted cell-free networks. Particularly, we adopted a stochastic programming method to cope with the CSI uncertainty by maximizing the expectation of the sum-rate, which guaranteed robust performance over the average. Accordingly, we formulated an average sum-rate maximization problem and proposed the efficient BCD-based RJD algorithm to solve it and finally obtain the locally optimal solutions in nearly closed-forms. 
Moreover, we further proved that the CSI uncertainty impacted the optimization of the active transmitting beamforming of APs, but surprisingly does not directly impact the optimization of the passive reflecting beamforming of IRSs.
Simulation results confirmed that deploying IRSs can considerably enhance the data rate compared with conventional cell-free networks. In addition, it was also shown that the proposed BCD-based RJD algorithm had a good convergence behavior and was robust against the CSI uncertainty.

\appendix

\subsection{Proof of Proposition \ref{proepec}}\label{appexp}
By substituting \eqref{equ4} into the expectation term, we have
\begin{align}\label{equ41}
    &\mathbb E\left\{\sum_{k=1}^{\mathcal K}{\operatorname{Tr}\left({\mathbf{\bar U}_k}\mathbf Y_k^{\operatorname{H}}\sum_{l=1}^{\mathcal L}{\sum_{i=1}^{\mathcal K}{{\mathbf {\bar H}_{l,k}^{\operatorname{H}}\mathbf W_{l,i}\mathbf W_{l,i}^{\operatorname{H}}\mathbf {\bar H}_{l,k}}\mathbf Y_{k}}}\right)}\right\}\notag\\
    =&\underbrace{\mathbb E\left\{\sum_{k=1}^{\mathcal K}{\operatorname{Tr}\left({\mathbf{\bar U}_k}\mathbf Y_k^{\operatorname{H}}\sum_{l=1}^{\mathcal L}\sum_{i=1}^{\mathcal K}{\mathbf {\bar D}_{l,k}^{\operatorname{H}}\mathbf W_{l,i}\mathbf W_{l,i}^{\operatorname{H}}\mathbf {\bar D}_{l,k}}\mathbf Y_k\right)}\right\}}_{\left(\operatorname{a}\right)}\notag\\
    &+\underbrace{\mathbb E\left\{\sum_{k=1}^{\mathcal K}{\operatorname{Tr}\left({\mathbf{\bar U}_k}\mathbf Y_k^{\operatorname{H}}\sum_{l=1}^{\mathcal L}\sum_{i=1}^{\mathcal K}{\mathbf {\bar G}_{k}^{\operatorname{H}}\boldsymbol{\Theta}\mathbf {\hat S}_l\mathbf W_{l,i}\mathbf W_{l,i}^{\operatorname{H}}\mathbf {\hat S}_l^{\operatorname{H}}\boldsymbol{\Theta}^{\operatorname{H}}\mathbf {\bar G}_{k}}\mathbf Y_k\right)}\right\}}_{\left(\operatorname{b}\right)}\notag\\
    &+\underbrace{\mathbb E\left\{\sum_{k=1}^{\mathcal K}{\operatorname{Tr}\left({\mathbf{\bar U}_k}\mathbf Y_k^{\operatorname{H}}\sum_{l=1}^{\mathcal L}\sum_{i=1}^{\mathcal K}{\mathbf {\hat G}_k^{\operatorname{H}}\boldsymbol{\Theta}\mathbf {\bar S}_{l}\mathbf W_{l,i}\mathbf W_{l,i}^{\operatorname{H}}\mathbf {\bar S}_{l}^{\operatorname{H}}\boldsymbol{\Theta}^{\operatorname{H}}\mathbf {\hat G}_k}\mathbf Y_k\right)}\right\}}_{\left(\operatorname{c}\right)}\notag\\
    &+\underbrace{\mathbb E\left\{\sum_{k=1}^{\mathcal K}{\operatorname{Tr}\left({\mathbf{\bar U}_k}\mathbf Y_k^{\operatorname{H}}\sum_{l=1}^{\mathcal L}\sum_{i=1}^{\mathcal K}{\mathbf {\bar G}_k^{\operatorname{H}}\boldsymbol{\Theta}\mathbf {\bar S}_{l}\mathbf W_{l,i}\mathbf W_{l,i}^{\operatorname{H}}\mathbf {\bar S}_{l}^{\operatorname{H}}\boldsymbol{\Theta}^{\operatorname{H}}\mathbf {\bar G}_k}\mathbf Y_k\right)}\right\}}_{\left(\operatorname{d}\right)}.
\end{align}
The above equation holds based on the fact that $\mathbb E\left\{\mathbf {\bar S}_{l,r}\right\}=\mathbb E\left\{\mathbf {\bar G}_{r,k}\right\}=\mathbb E\left\{\mathbf {\bar D}_{l,k}\right\}=\boldsymbol 0, \forall l,r,k$. 

In order to complete this proof, we make use of the following properties for any matrices of appropriate dimensions:
\begin{subequations}
\begin{align}
    \operatorname{Tr}\left(\mathbf A^{\operatorname{T}}\mathbf B\right)&=\operatorname{Vec}\left(\mathbf A\right)^{\operatorname{T}}\operatorname{Vec}\left(\mathbf B\right),\label{equpa}\\
    \operatorname{Tr}\left(\mathbf A\otimes\mathbf B\right)&=\operatorname{Tr}\left(\mathbf A\right)\operatorname{Tr}\left(\mathbf B\right),\label{equpb}\\
    \operatorname{Vec}\left(\mathbf A\mathbf B\mathbf C\right)&=\left(\mathbf C^{\operatorname{T}} \otimes \mathbf A\right)\operatorname{Vec}\left(\mathbf B\right),\label{equpc}\\
    \operatorname{Tr}\left(\mathbf A\mathbf M\mathbf B\mathbf M\right)&=\mathbf m^{\operatorname{T}}\left(\mathbf A\odot \mathbf B\right)\mathbf m,\label{equpd}\\
    \mathbf A \odot \mathbf I&=\operatorname{Diag}\left(\mathbf A_{1,1},\mathbf A_{2,2},\cdots,\mathbf A_{J,J}\right),\label{equpe}\\
    \mathbb E\left\{\mathbf x ^{\operatorname{T}}\mathbf A\mathbf x\right\}&=\operatorname{Tr}\left(\mathbf A\boldsymbol \Sigma\right)+\mathbf c ^{\operatorname{T}}\mathbf A\mathbf c,\label{equpf}
\end{align}
\end{subequations}
where $\mathbf M=\operatorname{Diag}\left(\mathbf m\right)$ and $\mathbf m=\operatorname{Vecd}\left(\mathbf M\right)$, and $\mathbf x$ is a stochastic vector with mean $\mathbf c$ and covariance $\boldsymbol \Sigma$. The proofs of above matrix manipulations can be found in \cite{IMM201203274} and \cite{2017Matrix}.

Then, in the following, we analysis the expectation terms $\left(\operatorname{a}\right)$, $\left(\operatorname{b}\right)$, $\left(\operatorname{c}\right)$, $\left(\operatorname{d}\right)$, respectively. First, we can express the expected value of the term $\left(\operatorname{a}\right)$ as 
\begin{subequations}\label{equa}
\begin{align}
   \left(\operatorname{a}\right)=&\mathbb E\left\{\sum_{k=1}^{\mathcal K}{\sum_{l=1}^{\mathcal L}{\sum_{i=1}^{\mathcal K}{\operatorname{Vec}\left(\mathbf {\bar D}_{l,k}\right)^{\operatorname{H}}\operatorname{Vec}\left(\mathbf W_{l,i}\mathbf W_{l,i}^{\operatorname{H}}\mathbf {\bar D}_{l,k}\mathbf Y_k{\mathbf{\bar U}_k}\mathbf Y_k^{\operatorname{H}}\right)}}}\right\}\label{equaa}\\
   =&\mathbb E\left\{\sum_{k=1}^{\mathcal K}{\sum_{l=1}^{\mathcal L}{\sum_{i=1}^{\mathcal K}{\operatorname{Vec}\left(\mathbf {\bar D}_{l,k}\right)^{\operatorname{H}}\left(\left(\mathbf Y_k{\mathbf{\bar U}_k}\mathbf Y_k^{\operatorname{H}}\right)^{\operatorname{T}}\otimes\left(\mathbf W_{l,i}\mathbf W_{l,i}^{\operatorname{H}}\right)\right)\operatorname{Vec}\left(\mathbf {\bar D}_{l,k}\right)}}}\right\}\label{equab}\\
   =&\sum_{k=1}^{\mathcal K}{\sum_{l=1}^{\mathcal L}{\sum_{i=1}^{\mathcal K}{\delta_{\mathbf D_{l,k}}^2\operatorname{Tr}\left(\left(\mathbf Y_k{\mathbf{\bar U}_k}\mathbf Y_k^{\operatorname{H}}\right)^{\operatorname{T}}\otimes\left(\mathbf W_{l,i}\mathbf W_{l,i}^{\operatorname{H}}\right)\right)}}}\label{equac}\\
   =&\sum_{k=1}^{\mathcal K}{\sum_{l=1}^{\mathcal L}{\sum_{i=1}^{\mathcal K}{\delta_{\mathbf D_{l,k}}^2\operatorname{Tr}\left(\mathbf Y_k{\mathbf{\bar U}_k}\mathbf Y_k^{\operatorname{H}}\right)\operatorname{Tr}\left(\mathbf W_{l,i}\mathbf W_{l,i}^{\operatorname{H}}\right)}}},\label{equad}
\end{align}
\end{subequations}
where \eqref{equaa}, \eqref{equab}, \eqref{equac}, and \eqref{equad} can be achieved by employing the properties \eqref{equpa}, \eqref{equpc}, \eqref{equpf}, and \eqref{equpb}, respectively.

Then, we can express the expected value of the term $\left(\operatorname{b}\right)$ as
\begin{subequations}\label{equb}
\begin{align}
    \left(\operatorname{b}\right)
    &=
    \sum_{k=1}^{\mathcal K}{\sum_{l=1}^{\mathcal L}{\sum_{i=1}^{\mathcal K}{\delta_{\mathbf G_{k}}^2\operatorname{Tr}\left(\mathbf Y_k{\mathbf{\bar U}_k}\mathbf Y_k^{\operatorname{H}}\right)\operatorname{Tr}\left(\boldsymbol{\Theta}\mathbf {\hat S}_l\mathbf W_{l,i}\mathbf W_{l,i}^{\operatorname{H}}\mathbf {\hat S}_l^{\operatorname{H}}\boldsymbol{\Theta}^{\operatorname{H}}\right)}}}\label{equba}\\
    &=\sum_{k=1}^{\mathcal K}{\sum_{l=1}^{\mathcal L}{\sum_{i=1}^{\mathcal K}{\alpha^2\delta_{\mathbf G_{k}}^2\operatorname{Tr}\left(\mathbf Y_k{\mathbf{\bar U}_k}\mathbf Y_k^{\operatorname{H}}\right)\operatorname{Tr}\left(\mathbf {\hat S}_l\mathbf W_{l,i}\mathbf W_{l,i}^{\operatorname{H}}\mathbf {\hat S}_l^{\operatorname{H}}\right)}}},\label{equbb}
\end{align}
\end{subequations}
where the proof of \eqref{equba} is the same as \eqref{equa}, and \eqref{equba} holds based on the property of the trace operator, i.e., $\operatorname{Tr}\left(\mathbf A\mathbf B\mathbf C\right)=\operatorname{Tr}\left(\mathbf B\mathbf C\mathbf A\right)=\operatorname{Tr}\left(\mathbf C\mathbf A\mathbf B\right)$ and $\boldsymbol\Theta^{\operatorname{H}}\boldsymbol\Theta=\alpha^2\mathbf I_{\mathcal R\mathcal N}$. Besides, \eqref{equbb} can also be proved by employing \eqref{equpd} and \eqref{equpe} as follows:
\begin{subequations}
\begin{align}
  \operatorname{Tr}\left(\boldsymbol{\Theta}\mathbf {\hat S}_l\mathbf W_{l,i}\mathbf W_{l,i}^{\operatorname{H}}\mathbf {\hat S}_l^{\operatorname{H}}\boldsymbol{\Theta}\right)&=\operatorname{Vecd}\left(\boldsymbol\Theta\right)^{\operatorname{H}}\left(\mathbf I\odot\left(\mathbf {\hat S}_l\mathbf W_{l,i}\mathbf W_{l,i}^{\operatorname{H}}\mathbf {\hat S}_l^{\operatorname{H}}\right)\right)\operatorname{Vecd}\left(\boldsymbol\Theta \right)\label{equbba}\\
  &=\operatorname{Vecd}\left(\boldsymbol\Theta\right)^{\operatorname{H}}\operatorname{Diag}\left(\mathbf {\hat S}_l\mathbf W_{l,i}\mathbf W_{l,i}^{\operatorname{H}}\mathbf {\hat S}_l^{\operatorname{H}}\right)\operatorname{Vecd}\left(\boldsymbol\Theta \right)\label{equbbb}\\
  &=\alpha^2\operatorname{Tr}\left(\mathbf {\hat S}_l\mathbf W_{l,i}\mathbf W_{l,i}^{\operatorname{H}}\mathbf {\hat S}_l^{\operatorname{H}}\right)\label{equbbc},
\end{align}
\end{subequations}
where \eqref{equbbc} holds based on the definition of the trace operator and $\boldsymbol\Theta_{n,n}^{\ast}\boldsymbol\Theta_{n,n}=\alpha^2, \forall n$.

Similarly, the expected value of the term $\left(\operatorname{c}\right)$ can be expressed as
\begin{align}
    \left(\operatorname{c}\right)&=\sum_{k=1}^{\mathcal K}{\sum_{l=1}^{\mathcal L}{\sum_{i=1}^{\mathcal K}{\alpha^2\delta_{\mathbf S_{l}}^2\operatorname{Tr}\left(\mathbf {\hat G}_k\mathbf Y_k{\mathbf{\bar U}_k}\mathbf Y_k^{\operatorname{H}}\mathbf {\hat G}_k^{\operatorname{H}}\right)\operatorname{Tr}\left(\mathbf W_{l,i}\mathbf W_{l,i}^{\operatorname{H}}\right)}}}.\label{equca}
\end{align}

Further, the expected value of the term $\left(\operatorname{d}\right)$ can be determined by
\begin{subequations}\label{eqd}
\begin{align}
        \left(\operatorname{d}\right)&=\sum_{k=1}^{\mathcal K}{\sum_{l=1}^{\mathcal L}{\sum_{i=1}^{\mathcal K}{\alpha^2\delta_{\mathbf G_{k}}^2\operatorname{Tr}\left(\mathbf Y_k{\mathbf{\bar U}_k}\mathbf Y_k^{\operatorname{H}}\right)\operatorname{Tr}\left(\mathbf {\bar S}_{l}\mathbf W_{l,i}\mathbf W_{l,i}^{\operatorname{H}}\mathbf {\bar S}_{l}^{\operatorname{H}}\right)}}}\\
        &=\sum_{k=1}^{\mathcal K}{\sum_{l=1}^{\mathcal L}{\sum_{i=1}^{\mathcal K}{\alpha^2\delta_{\mathbf S_{l}}^2\delta_{\mathbf G_{k}}^2\operatorname{Tr}\left(\mathbf Y_k{\mathbf{\bar U}_k}\mathbf Y_k^{\operatorname{H}}\right)\operatorname{Tr}\left(\mathbf W_{l,i}\mathbf W_{l,i}^{\operatorname{H}}\otimes \mathbf I_{\mathcal R\mathcal N}\right)}}}\label{equcba}\\
        &=\sum_{k=1}^{\mathcal K}{\sum_{l=1}^{\mathcal L}{\sum_{i=1}^{\mathcal K}{\mathcal R\mathcal N\alpha^2\delta_{\mathbf S_{l}}^2\delta_{\mathbf G_{k}}^2\operatorname{Tr}\left(\mathbf Y_k{\mathbf{\bar U}_k}\mathbf Y_k^{\operatorname{H}}\right)\operatorname{Tr}\left(\mathbf W_{l,i}\mathbf W_{l,i}^{\operatorname{H}}\right)}}},\label{equcb}
\end{align}
\end{subequations}
where \eqref{equcba} and \eqref{equcb} hold based on \eqref{equpf} and $\operatorname{Tr}\left(\mathbf I_{\mathcal R\mathcal N}\right)=\mathcal R\mathcal N$, respectively.

By combining \eqref{equa}-\eqref{eqd}, the expectation term in \eqref{equ41} can be simplified as
\begin{align}
&\mathbb E\left\{\sum_{k=1}^{\mathcal K}{\operatorname{Tr}\left({\mathbf{\bar U}_k}\mathbf Y_k^{\operatorname{H}}\sum_{l=1}^{\mathcal L}\sum_{i=1}^{\mathcal K}{\mathbf {\bar H}_{l,k}^{\operatorname{H}}\mathbf W_{l,i}\mathbf W_{l,i}^{\operatorname{H}}\mathbf {\bar H}_{l,k}}\mathbf Y_{k}\right)}\right\}\notag\\
    =&\sum_{k=1}^{\mathcal K}{\sum_{l=1}^{\mathcal L}{\sum_{i=1}^{\mathcal K}{\delta_{\mathbf D_{l,k}}^2\operatorname{Tr}\left(\mathbf Y_k{\mathbf{\bar U}_k}\mathbf Y_k^{\operatorname{H}}\right)\operatorname{Tr}\left(\mathbf W_{l,i}\mathbf W_{l,i}^{\operatorname{H}}\right)}}}\notag\\
    &+\sum_{k=1}^{\mathcal K}{\sum_{l=1}^{\mathcal L}{\sum_{i=1}^{\mathcal K}{\alpha^2\delta_{\mathbf G_{k}}^2\operatorname{Tr}\left(\mathbf Y_k{\mathbf{\bar U}_k}\mathbf Y_k^{\operatorname{H}}\right)\operatorname{Tr}\left(\mathbf {\hat S}_l\mathbf W_{l,i}\mathbf W_{l,i}^{\operatorname{H}}\mathbf {\hat S}_l^{\operatorname{H}}\right)}}}\notag\\
    &+\sum_{k=1}^{\mathcal K}{\sum_{l=1}^{\mathcal L}{\sum_{i=1}^{\mathcal K}{\alpha^2\delta_{\mathbf S_{l}}^2\operatorname{Tr}\left(\mathbf {\hat G}_k\mathbf Y_k{\mathbf{\bar U}_k}\mathbf Y_k^{\operatorname{H}}\mathbf {\hat G}_k^{\operatorname{H}}\right)\operatorname{Tr}\left(\mathbf W_{l,i}\mathbf W_{l,i}^{\operatorname{H}}\right)}}}\notag\\
    &+\sum_{k=1}^{\mathcal K}{\sum_{l=1}^{\mathcal L}{\sum_{i=1}^{\mathcal K}{\mathcal R\mathcal N\alpha^2\delta_{\mathbf S_{l}}^2\delta_{\mathbf G_{k}}^2\operatorname{Tr}\left(\mathbf Y_k{\mathbf{\bar U}_k}\mathbf Y_k^{\operatorname{H}}\right)\operatorname{Tr}\left(\mathbf W_{l,i}\mathbf W_{l,i}^{\operatorname{H}}\right)}}}.
\end{align}

This completes the proof of Proposition \ref{proepec}.

\subsection{Proof of Lemma \ref{lemmaconver}}\label{lemma2proof}
Let $\rho _i^{\operatorname{u+1}}$ denote the value of the objective function of $\rho ^{\operatorname{u+1}}$ after updating the $i$-th element and fixing the others $\mathcal {\widetilde N} -1$ elements of $\boldsymbol\theta$ in the $u$-th sub-iteration, we have
\begin{align}\label{eqlem}
    \rho^{\operatorname{u}}\le\rho _1^{\operatorname{u+1}}\le\rho _2^{\operatorname{u+1}}\cdots\le\rho _{\mathcal {\widetilde N}}^{\operatorname{u+1}}=\rho ^{\operatorname{u+1}},
\end{align}
which shows that the value of the objective function in problem $\mathcal P5$ achieved by Algorithm \ref{a2} increases monotonically. Meanwhile, we have
\begin{subequations}
\begin{align}
    -\boldsymbol\theta^{\operatorname{H}}\boldsymbol{\mathcal { Z}}\boldsymbol\theta \le -\alpha^2\mathcal{\widetilde N}\lambda_{\boldsymbol{\mathcal { Z}}}^{\min},\\
\operatorname{Re}\left\{\boldsymbol\theta^{\operatorname{H}}\boldsymbol\omega\right\} \le \alpha\sum_{n=1}^{\mathcal{\widetilde N}}\left|\boldsymbol\omega_n\right|,
\end{align}
\end{subequations}
where $\lambda_{\boldsymbol{\mathcal { Z}}}^{\min}$ is the minimum eigenvalue of $\boldsymbol{\mathcal { Z}}$. The above inequalities yield that the optimal objective value is upper-bounded by a finite value. Therefore, Algorithm \ref{a2} is guaranteed  to converge.
This completes the proof of Lemma \ref{lemmaconver}.


\ifCLASSOPTIONcaptionsoff
  \newpage
\fi

\bibliography{ref}
\end{document}